\documentclass[10pt,conference]{IEEEtran}
\usepackage{fixltx2e}
\usepackage{cite}
\usepackage{url}
\usepackage{mathrsfs}
\usepackage{color}
\usepackage{float}
\usepackage[caption=false]{subfig}

\ifCLASSINFOpdf
   \usepackage[pdftex]{graphicx}
   \graphicspath{{Figs/}}
   \DeclareGraphicsExtensions{.pdf,.jpeg,.png}
\else
\fi

\usepackage[cmex10]{amsmath}
\usepackage{amsmath}
\usepackage{amssymb}
\usepackage{amsthm}
\usepackage{amsfonts}
\usepackage{bm}
\usepackage{xfrac}
\usepackage{empheq}
\usepackage[normalem]{ulem} 
\usepackage{soul} 
\usepackage{mathtools}

\newtheorem{theorem}{Theorem}

\newtheorem{example}{Example}
\newtheorem{proposition}{Proposition}
\newtheorem{lemma}{Lemma}

\newtheorem{corollary}{Corollary}
\newtheorem{remark}{Remark}
\theoremstyle{definition}
\newtheorem{definition}{Definition}

\usepackage[a4paper,bindingoffset=0.2in,%
left=1in,right=1in,top=1in,bottom=1in,%
footskip=.25in]{geometry}
\graphicspath{{figs/}}

\begin{document}
	\newgeometry{left=0.7in,right=0.7in,top=.5in,bottom=1in}
	\title{Bounds for Privacy-Utility Trade-off with Per-letter Privacy Constraints and Non-zero Leakage}
\vspace{-5mm}
\author{
		\IEEEauthorblockN{Amirreza Zamani, Tobias J. Oechtering, Mikael Skoglund \vspace*{0.5em}
			\IEEEauthorblockA{\\
                              Division of Information Science and Engineering, KTH Royal Institute of Technology \\
				Email: \protect amizam@kth.se, oech@kth.se, skoglund@kth.se }}
		}
	\maketitle

%
\begin{abstract}
	An information theoretic privacy mechanism design problem for two scenarios is studied where the private data is either hidden or observable. In each scenario, privacy leakage constraints are considered using two different measures. In these scenarios the private data is hidden or observable.  
	In the first scenario, 
	an agent observes useful data $Y$ that is correlated with private data $X$, and wishes to disclose the useful information to a user. 
	A privacy mechanism is designed to generate disclosed data $U$ which maximizes the revealed information about $Y$ while satisfying a per-letter privacy constraint. 
	In the second scenario, the agent has additionally access to the private data. 
	First, the Functional Representation Lemma and Strong Functional Representation Lemma are extended by relaxing the independence condition to find a lower bound considering the second scenario. 
	Next, lower bounds as well as upper bounds on privacy-utility trade-off are derived for both scenarios. 
	In particular, for the case where $X$ is deterministic function of $Y$, we show that our upper and lower bounds are asymptotically optimal considering the first scenario.
\end{abstract}
\section{Introduction}

The privacy mechanism design problem is recently receiving increased attention in information theory
\cite{ makhdoumi, issa, Calmon2,yamamoto, sankar,borz, gun,khodam,Khodam22,kostala, dwork1, calmon4, issajoon, asoo, Total, issa2,zamani2022bounds,kosenaz}. 
Specifically, in \cite{makhdoumi}, the concept of a privacy funnel is introduced, where the privacy utility trade-off has been studied considering a distortion measure for utility and the log-loss as privacy measure. In \cite{issa}, the concept of maximal leakage has been introduced and some bounds on the privacy utility trade-off have been derived. 
Fundamental limits of the privacy utility trade-off measuring the leakage using estimation-theoretic guarantees are studied in \cite{Calmon2}.
A related secure source coding problem is studied in \cite{yamamoto}.

In both \cite{yamamoto} and \cite{sankar}, the privacy-utility trade-offs considering expected distortion and equivocation as a measures of utility and privacy are studied.
The problem of privacy-utility trade-off considering mutual information both as measures of utility and privacy given the Markov chain $X-Y-U$ is studied in \cite{borz}. Under the perfect privacy assumption it is shown that the privacy mechanism design problem can be reduced to a linear program. This has been extended in \cite{gun} considering the privacy utility trade-off with a rate constraint on the disclosed data.
Moreover, in \cite{borz}, it has been shown that information can be only revealed if the kernel (leakage matrix) between useful data and private data is not invertible. In \cite{khodam}, we generalize \cite{borz} by relaxing the perfect privacy assumption allowing some small bounded leakage. More specifically, we design privacy mechanisms with a per-letter privacy criterion considering an invertible kernel where a small leakage is allowed. We generalized this result to a non-invertible leakage matrix in \cite{Khodam22}.\\
In this paper, random variable (RV) $Y$ denotes the useful data and is correlated with the private data denoted by RV $X$. Furthermore, RV $U$ describes the disclosed data. Two scenarios are considered in this work, where in both scenarios, an agent wants to disclose the useful information to a user as shown in Fig.~\ref{ISITsys}. In the first scenario, the agent observes $Y$ and has not directly access to $X$, i.e., the private data is hidden. The goal is to design $U$ based on $Y$ that reveals as much information as possible about $Y$ and satisfies a privacy criterion. 
In the second scenario, the agent has access to both $X$ and $Y$ and can design $U$ based on $(X,Y)$ to release as much information as possible about $Y$ while satisfying the bounded leakage constraint.
In both scenarios we consider two different per-letter privacy criterion.  \\
\begin{figure}[]
	\centering
	\includegraphics[scale = .15]{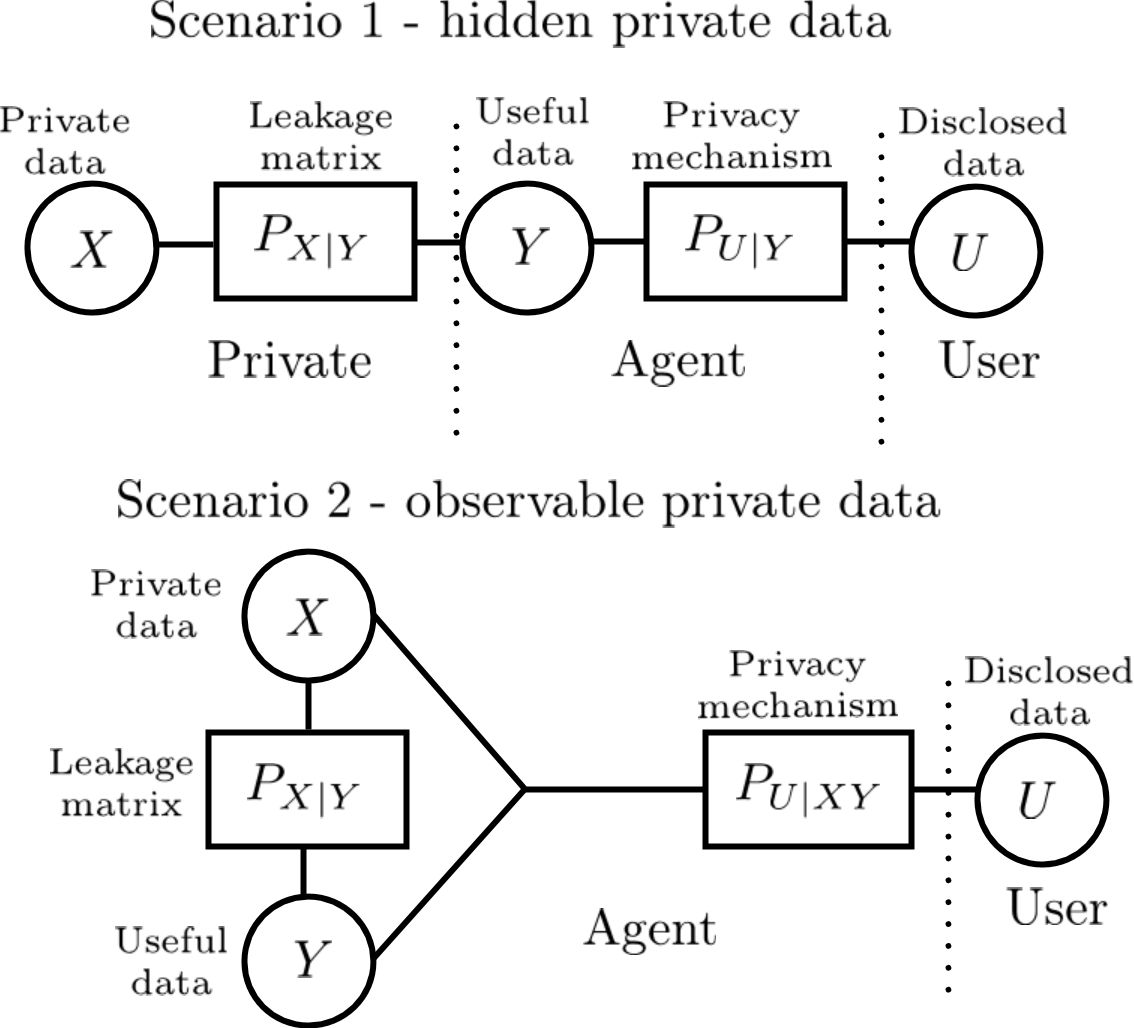}
	\caption{In the first scenario the agent has only access to $Y$ and in the second scenario the agent has additionally access to $X$.}
	\label{ISITsys}
\end{figure}
In \cite{kostala}, by using the Functional Representation Lemma bounds on privacy-utility trade-off for the two scenarios are derived. These results are derived under the perfect secrecy assumption, i.e., no leakages are allowed. The bounds are tight when the private data is a deterministic function of the useful data. In \cite{zamani2022bounds}, we generalize the privacy problems considered in \cite{kostala} by relaxing the perfect privacy constraint and allowing some leakages. More specifically, we considered bounded mutual information, i.e., $I(U;X)\leq \epsilon$ for privacy leakage constraint. Furthermore, in the special case of perfect privacy we found a new upper bound for the perfect privacy function by using the \emph{excess functional information} introduced in \cite{kosnane}. It has been shown that this new bound generalizes the bound in \cite{kostala}. Moreover, we have shown that the bound is tight when $|\mathcal{X}|=2$.

In \cite{zamani2022bounds}, we have used mutual information for measuring the privacy leakage, however in the present work, for each scenario we use two different per letter privacy constraints. As argued in \cite{Khodam22}, it is more desirable to protect the private data individually and not on average. By using an average constraint, a data point can exist which leaks more than average threshold.

In this work, we first derive similar lemmas as \cite[Lemma~3]{zamani2022bounds} and \cite[Lemma~4]{zamani2022bounds} where we have extended the Functional Representation Lemma and the Strong Functional Representation Lemma considering bounded leakage, i.e., $I(U;X)\leq \epsilon$, instead of independent $X$ and $U$. In this paper, we derive similar results considering per-letter privacy constraint rather than bounded mutual information. Using these lemmas we find a lower bound for the privacy-utility trade-off in the second scenario with first per letter leakage constraint. Furthermore, we provide bounds for three other problems and study a special case where $X$ is a deterministic function of $Y$. We show that the obtained upper and lower bounds in the first scenario are asymptotically optimal when $X$ is a deterministic function of $Y$. Finally, we evaluate the bounds in a numerical example.       

\section{system model and Problem Formulation} \label{sec:system}
Let $P_{XY}$ denote the joint distribution of discrete random variables $X$ and $Y$ defined on finite alphabets $\cal{X}$ and $\cal{Y}$ with $|\mathcal{X}|<|\mathcal{Y}|$.
We represent $P_{XY}$ by a matrix defined on $\mathbb{R}^{|\mathcal{X}|\times|\mathcal{Y}|}$ and 
marginal distributions of $X$ and $Y$ by vectors $P_X$ and $P_Y$ defined on $\mathbb{R}^{|\mathcal{X}|}$ and $\mathbb{R}^{|\mathcal{Y}|}$ given by the row and column sums of $P_{XY}$. 
We assume that each element in vectors $P_X$ and $P_Y$ is non-zero. Furthermore, 
we represent the leakage matrix $P_{X|Y}$ by a matrix defined on $\mathbb{R}^{|\mathcal{X}|\times|\cal{Y}|}$, which is assumed to be of full rank. Furthermore, for given $u\in \mathcal{U}$, $P_{X,U}(\cdot,u)$ and $P_{X|U}(\cdot|u)$ defined on $\mathbb{R}^{|\mathcal{X}|}$ are distribution vectors with elements $P_{X,U}(x,u)$ and $P_{X|U}(x|u)$ for all $x\in\cal X$ and $u\in \cal U$. 
The relation between $U$ and $Y$ is described by the kernel $P_{U|Y}$ defined on $\mathbb{R}^{|\mathcal{U}|\times|\mathcal{Y}|}$, furthermore, the relation between $U$ and the pair $(Y,X)$ is described by the kernel $P_{U|Y,X}$ defined on $\mathbb{R}^{|\mathcal{U}|\times|\mathcal{Y}|\times|\mathcal{X}|}$.

 The privacy mechanism design problems for the two scenarios can be stated as follows 
\begin{align}
g_{\epsilon}^1(P_{XY})&=\sup_{\begin{array}{c} 
	\substack{P_{U|Y}:X-Y-U\\ \ d(P_{X,U}(\cdot,u),P_XP_{U}(u))\leq\epsilon,\ \forall u}
	\end{array}}I(Y;U),\label{main2}\\
h_{\epsilon}^1(P_{XY})&=\sup_{\begin{array}{c} 
	\substack{P_{U|Y,X}: d(P_{X,U}(\cdot,u),P_XP_{U}(u))\leq\epsilon,\ \forall u}
	\end{array}}I(Y;U),\label{main1}\\
g_{\epsilon}^2(P_{XY})&=\sup_{\begin{array}{c} 
	\substack{P_{U|Y}:X-Y-U\\ \ d(P_{X|U}(\cdot|u),P_X)\leq\epsilon,\ \forall u}
	\end{array}}I(Y;U),\label{main22}\\
h_{\epsilon}^2(P_{XY})&=\sup_{\begin{array}{c} 
	\substack{P_{U|Y,X}: d(P_{X|U}(\cdot|u),P_X)\leq\epsilon,\ \forall u}
	\end{array}}I(Y;U),\label{main12}
\end{align} 
where $d(P,Q)$ corresponds to the total variation distance between two distributions $P$ and $Q$, i.e., $d(P,Q)=\sum_x |P(x)-Q(x)|$. 
 The functions $h_{\epsilon}^1(P_{XY})$ and $h_{\epsilon}^2(P_{XY})$ are used when the privacy mechanism has access to both the private data and the useful data. The functions $g_{\epsilon}^1(P_{XY})$ and $g_{\epsilon}^2(P_{XY})$ are used when the privacy mechanism has only access to the useful data. In this work, the privacy constraints used in \eqref{main2} and \eqref{main22}, i.e., $d(P_{X,U}(\cdot|u),P_XP_U(u))\leq\epsilon,\ \forall u,$ and $d(P_{X|U}(\cdot|u),P_X)\leq\epsilon,\ \forall u,$ are called the \emph{strong privacy criterion 1} and the \emph{strong privacy criterion 2}. We call them strong since they are per-letter privacy constraints, i.e., they must hold for every $u\in\cal U$. The difference between the two privacy constraints in this work is the weight $P_U(u)$, which we later show that it enables us to use extended versions of the Functional Representation Lemma and Strong Functional Representation Lemma to find lower bounds considering the second scenario. 
 \begin{remark}
	\normalfont
	We have used the leakage constraint $d(P_{X|U}(\cdot|u),P_X)\leq\epsilon,\ \forall u$ in \cite{Khodam22}, where we called it the \emph{strong $\ell_1$-privacy criterion}. 
\end{remark}  
 \begin{remark}
 	\normalfont
 	For $\epsilon=0$, both \eqref{main2} and \eqref{main22} lead to the perfect privacy problem studied in \cite{borz}. It has been shown that for a non-invertible leakage matrix $P_{X|Y}$, $g_0(P_{XY})$ can be obtained by a linear program.
 \end{remark}
 \begin{remark}
 	\normalfont
 	For $\epsilon=0$, both \eqref{main1} and \eqref{main12} lead to the secret-dependent perfect privacy function $h_0(P_{XY})$, studied in \cite{kostala}, where upper and lower bounds on $h_0(P_{XY})$ have been derived. In \cite{zamani2022bounds}, we have strengthened these bounds.
 \end{remark}
 \begin{remark}
 	\normalfont
 	The privacy problem defined in \eqref{main22} has been studied in \cite{Khodam22} where we provide a lower bound on $g_{\epsilon}^2(P_{XY})$ using the information geometry concepts. Furthermore, we have shown that with out loss of optimality it is sufficient to assume $|\mathcal{U}|\leq |\mathcal{Y}|$ so that it is ensured that the supremum can be achieved. 
 \end{remark}
 Intuitively, for small $\epsilon$, both privacy constraints mean that $X$ and $U$ are almost independent. As we discussed in \cite{Khodam22}, closeness of $P_{X|U}(\cdot|u)$ and $P_X$ allows us to approximate $g_{\epsilon}^2(P_{XY})$ with a series expansion and find a lower bound. In this work we show that by using a similar methodology, we can approximate $g_{\epsilon}^1(P_{XY})$ exploiting the closeness of $P_{X,U}(\cdot,u)$ and $P_XP_U(u)$. This provides us a lower bound for $g_{\epsilon}^1(P_{XY})$. Next, we study some properties of the strong privacy criterion 1 and the strong privacy criterion 2. To this end recall that the \emph{linkage inequality} is the property that if $\cal L$ measures the privacy leakage between two random variables and the Markov chain $X-Y-U$ holds then we have $\mathcal{L}(X;U)\leq\mathcal{L}(Y;U)$. Since the strong privacy criterion 1 and the strong privacy criterion 2 are per letter constraints we define $\mathcal{L}^1(X;U=u)\triangleq \left\lVert P_{X|U}(\cdot|u)-P_X \right\rVert_1$, $\mathcal{L}^1(Y;U=u)\triangleq \left\lVert P_{Y|U}(\cdot|u)-P_Y \right\rVert_1$, $\mathcal{L}^2(X;U=u)\triangleq \left\lVert P_{X,U}(\cdot,u)-P_XP_U(u) \right\rVert_1$, $\mathcal{L}^2(Y;U=u)\triangleq \left\lVert P_{Y,U}(\cdot,u)-P_YP_U(u) \right\rVert_1$.
 	\begin{proposition}
 	The strong privacy criterion 1 and the strong privacy criterion 2 satisfy the linkage inequality. Thus, for each $u\in\mathcal{U}$ we have $\mathcal{L}^1(X;U=u)\leq\mathcal{L}^1(Y;U=u)$ and $\mathcal{L}^2(X;U=u)\leq\mathcal{L}^2(Y;U=u)$.	
 	\end{proposition}
 \begin{proof}
 	The proof is provided in Appendix A.
 \end{proof}
 	As discussed in \cite[page 4]{Total}, one benefit of the linkage inequality is to keep the privacy in layers of private information which is discussed in the following. Assume that the Markov chain $X-Y-U$ holds and distribution of $X$ is not known. If we can find $\tilde{X}$ such that $X-\tilde{X}-Y-U$ holds and distribution of $\tilde{X}$ is known then by the linkage inequality we can conclude $\mathcal{L}(X;U=u)\leq \mathcal{L}(\tilde{X};U=u)$. In other words, if the framework is designed for $\tilde{X}$, then a privacy constraint on $\tilde{X}$ leads to the constraint on $X$, i.e., provides an upper bound for any pre-processed RV $X$. To have the Markov chain $X-\tilde{X}-Y-U$ consider the scenario where $\tilde{X}$ is the private data and $X$ is a function of private data which is not known. For instance let $\tilde{X}=(X_1,X_2,X_3)$ and $X=X_1$. Thus, the mechanism that is designed based on $\tilde{X}-Y-U$ preserves the leakage constraint on $X$ and $U$. As pointed out in \cite[Remark 2]{Total}, among all the $L^p$-norms ($p\geq 1$), only the $\ell_1$ norm satisfies the linkage inequality. Next, given a leakage measure $\mathcal{L}$ and let the Markov chain $X-Y-U$ hold, if we have  $\mathcal{L}(X;U)\leq \mathcal{L}(X;Y)$, then we say that the \emph{post processing inequality} holds. In this work we use $\mathcal{L}^1(X;U)= \sum_u P_U(u)\mathcal{L}^1(X;U=u)$, $\mathcal{L}^2(X;U)= \sum_u \mathcal{L}^2(X;U=u)$ and $\mathcal{L}^1(Y;U)=\sum_u P_U(u)\mathcal{L}^1(Y;U=u)$, $\mathcal{L}^2(Y;U)=\sum_u \mathcal{L}^2(Y;U=u)$.
 	\begin{proposition}
 			The average of strong privacy constraints 1 and 2 with weights $1$ and $P_U(u)$, respectively, satisfy the post-processing inequality, i.e., we have $\mathcal{L}^1(X;U)\leq\mathcal{L}^1(Y;U)$ and $\mathcal{L}^2(X;U)\leq\mathcal{L}^2(Y;U)$.	  
 	\end{proposition}
 \begin{proof}
 	The proof is same as proof of \cite[Theorem 3]{Total} which is based on the convexity of $\ell_1$-norm.
 \end{proof}
\begin{proposition}
	The strong privacy criterion 1 and 2 result in bounded inference threat that is modeled in \cite{Calmon1}.
\end{proposition}
\begin{proof}
	The strong privacy criterion 1 and 2 lead to a bounded on average constraint $\sum_u P_U(u)\left\lVert P_{X|U=u}\!-\!P_X \right\rVert_1=2TV(X;U)\leq \epsilon$, where $TV(.|.)$ corresponds to the total variation. Thus, using \cite[Theorem 4]{Total}, we conclude that inference threats are bounded.
\end{proof}
 	Another property of $\ell_1$ distance is the relation between the $\ell_1$-norm and probability of error in a hypothesis test. As argued in \cite[Remark~6.5]{polyanskiy2014lecture}, for the binary hypothesis test with $H_0:X\sim P$ and $H_1:X\sim Q$,
 	the expression $1-TV(P, Q)$ is the sum of false alarm and missed detection probabilities. Thus, we have $TV(P,Q)=1-2P_e$, where $P_e$ is the error probability (the probability that we can not decide the right distribution for $X$). To see a benefit, consider the scenario where we want to decide whether $X$ and $U$ are independent or correlated. Thus, let $P=P_{X,U}$, $Q=P_{X}P_{U}$, $H_0:X,U\sim P$ and $H_1:X,U\sim Q$. We have
 	\begin{align*}
 	TV(P_{X,U};P_{X}P_{U})&=\frac{1}{2}\sum_u P_U(u)\left\lVert P_{X|U=u}-P_X \right\rVert_1\!\\&\leq \frac{1}{2}\epsilon.
 	\end{align*}
 	Thus, by increasing the leakage, the error of probability decreases.  \\
 	Finally, if we use $\ell_1$ distance as privacy leakage, after approximating $g_{\epsilon}^1(P_{XY})$ and $g_{\epsilon}^2(P_{XY})$, we face linear program problems in the end, which are much easier to handle.
 \section{Main Results}\label{sec:resul}
 In this section, we first introduce similar lemmas as \cite[Lemma~3]{zamani2022bounds} and \cite[Lemma~4]{zamani2022bounds}, where we have replaced mutual information, i.e., $I(U;X)=\epsilon$, with a per letter constraint. In the remaining part of this work $d(\cdot,\cdot)$ corresponds to the total variation distance, i.e., $d(P,Q)=\sum_x |P(x)-Q(x)|$. 
 \begin{lemma}\label{lemma11}
 	For any $0\leq\epsilon< \sqrt{2I(X;Y)}$ and any pair of RVs $(X,Y)$ distributed according to $P_{XY}$ supported on alphabets $\mathcal{X}$ and $\mathcal{Y}$ where $|\mathcal{X}|$ is finite and $|\mathcal{Y}|$ is finite or countably infinite, there exists a RV $U$ supported on $\mathcal{U}$ such that $X$ and $U$ satisfy the strong privacy criterion 1, i.e., we have
 	\begin{align}\label{c11}
 	d(P_{X,U}(\cdot,u),P_XP_{U}(u))\leq\epsilon,\ \forall u,
 	\end{align}
 	$Y$ is a deterministic function of $(U,X)$, i.e., we have
 	\begin{align}
 	H(Y|U,X)=0,\label{c21}
 	\end{align}
 	and 
 	\begin{align}
 	|\mathcal{U}|\leq |\mathcal{X}|(|\mathcal{Y}|-1)+1.\label{c31}
 	\end{align}
 \end{lemma}
\begin{proof}
The proof is provided in Appendix B.	
\end{proof}
\begin{lemma}\label{lemma22} 
	For any $0\leq\epsilon< \sqrt{2I(X;Y)}$ and pair of RVs $(X,Y)$ distributed according to $P_{XY}$ supported on alphabets $\mathcal{X}$ and $\mathcal{Y}$ where $|\mathcal{X}|$ is finite and $|\mathcal{Y}|$ is finite or countably infinite with $I(X,Y)< \infty$, there exists a RV $U$ supported on $\mathcal{U}$ such that $X$ and $U$ satisfy the strong privacy criterion 1, i.e., we have
	\begin{align*}
	d(P_{X,U}(\cdot,u),P_XP_{U}(u))\leq\epsilon,\ \forall u,
	\end{align*}
	$Y$ is a deterministic function of $(U,X)$, i.e., we have 
	\begin{align*}
	H(Y|U,X)=0,
	\end{align*}
	$I(X;U|Y)$ can be  upper bounded as follows 
	\begin{align}
	I(X;U|Y)\!\leq \alpha H(X|Y)\!+\!(1-\alpha)\!\left[ \log(I(X;Y)+1)+4\right],\label{bala}
	\end{align}
	and 
	$
	|\mathcal{U}|\leq \left[|\mathcal{X}|(|\mathcal{Y}|-1)+2\right]\left[|\mathcal{X}|+1\right],
	$
	where $\alpha =\frac{\epsilon^2}{2H(X)}$.
\end{lemma}
\begin{proof}
	Let $U$ be found by ESFRL as in \cite[Lemma~4]{zamani2022bounds}, where we let the leakage be $\frac{\epsilon^2}{2}$. The first constraint in this statement can be obtained by using the same proof as Lemma \ref{lemma11} and \eqref{bala} can be derived using \cite[Lemma~4]{zamani2022bounds}. 
\end{proof}
In the next proposition we find a lower bound on $h_{\epsilon}^1(P_{XY})$ using Lemma~1 and Lemma~2. 
\begin{proposition}\label{prop111}
	For any $0\leq \epsilon< \sqrt{2I(X;Y)}$ and pair of RVs $(X,Y)$ distributed according to $P_{XY}$ supported on alphabets $\mathcal{X}$ and $\mathcal{Y}$ we have
	\begin{align}\label{prop12}
	h_{\epsilon}^1(P_{XY})\geq \max\{L_{h^1}^{1}(\epsilon),L_{h^1}^{2}(\epsilon)\},
	\end{align}
	where
	\begin{align*}
	L_{h_1}^{1}(\epsilon) &= H(Y|X)-H(X|Y)+\frac{\epsilon^2}{2},\\
	L_{h_1}^{2}(\epsilon) &= H(Y|X)-\alpha H(X|Y)+\frac{\epsilon^2}{2}\\&\ -(1-\alpha)\left( \log(I(X;Y)+1)+4 \right),\\
	\end{align*}
	with $\alpha=\frac{\epsilon^2}{2H(X)}$.
\end{proposition}
\begin{proof}
	For deriving $L_{h_1}^{1}(\epsilon)$ let $U$ be produced by Lemma \ref{lemma11}. Thus, $I(X;U)=\frac{\epsilon^2}{2}$ and $U$ satisfies \eqref{c11} and \eqref{c21}. We have
	\begin{align*}
	h_{\epsilon}^1(P_{XY})&\geq\\
	I(U;Y)&=I(X;U)\!+\!H(Y|X)\!-\!I(X;U|Y)\!-\!H(Y|X,U)\\&=\frac{\epsilon^2}{2}+H(Y|X)-H(X|Y)+H(X|Y,U)\\ &\geq \frac{\epsilon^2}{2}+H(Y|X)-H(X|Y).
	\end{align*}
	Next for deriving $L_{h_1}^{2}(\epsilon)$ let $U$ be produced by Lemma \ref{lemma22}. Hence, $I(X;U)=\frac{\epsilon^2}{2}$ and $U$ satisfies \eqref{c11}, \eqref{c21} and \eqref{bala}. We obtain
	\begin{align*}
		h_{\epsilon}^1(P_{XY})&\geq I(U;Y) = \frac{\epsilon^2}{2}+H(Y|X)-I(X;U|Y)\\ &\geq \frac{\epsilon^2}{2}+H(Y|X)-\alpha H(X|Y)\\&-(1-\alpha)\left( \log(I(X;Y)+1)+4 \right).
	\end{align*}
\end{proof}
In the next section, we provide a lower bound on $g_{\epsilon}^1(P_{XY})$ by following the same approach as in \cite{Khodam22}. For more details about the proofs and steps of approximation see \cite[Section III]{Khodam22}.
\subsection{Lower bound on $g_{\epsilon}^1(P_{XY})$}
In \cite{Khodam22}, we show that $g_{\epsilon}^2(P_{XY})$ can be approximated by a linear program. Using this result we can derive a lower bound for $g_{\epsilon}^2(P_{XY})$. In this part, we follow a similar approach to approximate $g_{\epsilon}^2(P_{XY})$ which results in a lower bound. Similar to \cite{Khodam22}, for sufficiently small $\epsilon$, by using the leakage constraint in $g_{\epsilon}^2(P_{XY})$, i.e., the strong privacy criterion 1, we can rewrite the distribution $P_{X,U}(\cdot,u)$ as a perturbation of $P_XP_U(u)$. Thus, for any $u$ we can write $P_{X,U}(\cdot,u)=P_XP_U(u)+\epsilon J_u$, where $J_u\in \mathbb{R}^{|\mathcal{X}|}$ is a perturbation vector and satisfies the following properties:
\begin{align}
\bm{1}^T\cdot J_u&=0,\ \forall u, \label{koon1}\\
\sum_u J_u&=\bm{0}\in \mathbb{R}^{|\mathcal{X}|},\label{koon2}\\
 \bm{1}^T\cdot |J_u|&\leq 1,\ \forall u, \label{koon3}
\end{align} 
where $|\cdot|$ corresponds to the absolute value of the vector. 
We define matrix $M\in \mathbb{R}^{|\mathcal{X}|\times|\mathcal{Y}|}$ which is used in the remaining part as follows: Let $V$ be the matrix of right eigenvectors of $P_{X|Y}$, i.e., $P_{X|Y}=U\Sigma V^T$ and $V=[v_1,\ v_2,\ ... ,\ v_{|\mathcal{Y}|}]$, then $M$ is defined as
\begin{align*}
M \triangleq \left[v_1,\ v_2,\ ... ,\ v_{|\mathcal{X}|}\right]^T.  
\end{align*}  
Similar to \cite[Proposition~2]{Khodam22}, we have the following result.
\begin{proposition}\label{prop222}
	In \eqref{main2}, it suffices to consider $U$ such that $|\mathcal{U}|\leq|\mathcal{Y}|$. Since the supremum in \eqref{main2} is achieved, we can replace the supremum by the maximum.
\end{proposition}
\begin{proof}
	The proof follows the similar lines as proof of \cite[Proposition~2]{Khodam22}. The only difference is that the new convex and compact set is as follows
	\begin{align*}
	\Psi\!=\!\left\{\!y\in\mathbb{R}_{+}^{|\mathcal{Y}|}|My\!=\!MP_Y\!+\!\frac{\epsilon}{P_U(u)} M\!\!\begin{bmatrix}
	P_{X|Y_1}^{-1}J_u\\0
	\end{bmatrix}\!\!,J_u\in\mathcal{J} \!\right\}\!,
	\end{align*}
	where $\mathcal{J}=\{J\in\mathbb{R}^{|\mathcal{X}|}_{+}|\left\lVert J\right\rVert_1\leq 1,\ \bm{1}^{T}\cdot J=0\}$ and $\mathbb{R}_{+}$ corresponds to non-negative real numbers. Only non-zero weights $P_U(u)$ are considered since in the other case the corresponding $P_{Y|U}(\cdot|u)$ does not appear in $H(Y|U)$. 
\end{proof}
\begin{lemma}\label{madar1}
	If the Markov chain $X-Y-U$ holds, for sufficiently small $\epsilon$ and every $u\in\mathcal{U}$, the vector $P_{Y|U}(\cdot|u)$ lies in the following convex polytope
	\begin{align*}
	\mathbb{S}_{u} = \left\{y\in\mathbb{R}_{+}^{|\mathcal{Y}|}|My=MP_Y+\frac{\epsilon}{P_U(u)} M\begin{bmatrix}
	P_{X|Y_1}^{-1}J_u\\0
	\end{bmatrix}\right\},
	\end{align*}
	where $J_u$ satisfies \eqref{koon1}, \eqref{koon2} and \eqref{koon3}. Furthermore, $P_U(u)>0$, otherwise $P_{Y|U}(\cdot|u)$ does not appear in $I(Y;U)$.
\end{lemma}
\begin{proof}
	Using the Markov chain $X-Y-U$, we have
	\begin{align*}
	P_{X|U=u}-P_X=P_{X|Y}[P_{Y|U=u}-P_Y]=\epsilon \frac{J_u}{P_U(u)}.
	\end{align*}
	Thus, by following the similar lines as \cite[Lemma~2]{Khodam22} and using the properties of Null($M$) as \cite[Lemma~1]{Khodam22}, we have
	\begin{align*}
	MP_{Y|U}(\cdot|u)=MP_Y+\frac{\epsilon}{P_U(u)} M\begin{bmatrix}
	P_{X|Y_1}^{-1}J_u\\0
	\end{bmatrix}.
	\end{align*} 
\end{proof}
By using the same arguments as \cite[Lemma~3]{Khodam22}, it can be shown that any vector inside $\mathbb{S}_{u}$ is a standard probability vector. Thus, by using \cite[Lemma~3]{Khodam22} and Lemma~2 we have following result.
\begin{theorem}
		We have the following equivalency 
	\begin{align}\label{equii}
	\min_{\begin{array}{c} 
		\substack{P_{U|Y}:X-Y-U\\ d(P_{X,U}(\cdot,u),P_XP_U(u))\leq\epsilon,\ \forall u\in\mathcal{U}}
		\end{array}}\! \! \! \!\!\!\!\!\!\!\!\!\!\!\!\!\!\!\!H(Y|U) =\!\!\!\!\!\!\!\!\! \min_{\begin{array}{c} 
		\substack{P_U,\ P_{Y|U=u}\in\mathbb{S}_u,\ \forall u\in\mathcal{U},\\ \sum_u P_U(u)P_{Y|U=u}=P_Y,\\ J_u \text{satisfies}\ \eqref{koon1},\ \eqref{koon2},\ \text{and}\ \eqref{koon3}}
		\end{array}} \!\!\!\!\!\!\!\!\!\!\!\!\!\!\!\!\!\!\!H(Y|U).
	\end{align}
\end{theorem}  
Furthermore, similar to \cite[Prpoposition~3]{Khodam22}, it can be shown that the minimum of $H(Y|U)$ occurs at the extreme points of the sets $\mathbb{S}_{u}$, i.e., for each $u\in \mathcal{U}$, $P_{Y|U}^*(\cdot|u)$ that minmizes $H(Y|U)$ must belong to the extreme points of $\mathbb{S}_{u}$. To find the extreme points of $\mathbb{S}_{u}$ let $\Omega$ be the set of indices which correspond to $|\mathcal{X}|$ linearly independent columns of $M$, i.e., $|\Omega|=|\mathcal{X}|$ and $\Omega\subset \{1,..,|\mathcal{Y}|\}$. Let $M_{\Omega}\in\mathbb{R}^{|\mathcal{X}|\times|\mathcal{X}|}$ be the submatrix of $M$ with columns indexed by the set $\Omega$. Assume that $\Omega = \{\omega_1,..,\omega_{|\mathcal{X}|}\}$, where $\omega_i\in\{1,..,|\mathcal{Y}|\}$ and all elements are arranged in an increasing order. The $\omega_i$-th element of the extreme point $V_{\Omega}^*$ can be found as $i$-th element of $M_{\Omega}^{-1}(MP_Y+\frac{\epsilon}{P_U(u)} M\begin{bmatrix}
P_{X|Y_1}^{-1}J_u\\0\end{bmatrix})$, i.e., for $1\leq i \leq |\mathcal{X}|$ we have
\begin{align}\label{defin1}
V_{\Omega}^*(\omega_i)= \left(M_{\Omega}^{-1}MP_Y+\frac{\epsilon}{P_U(u)} M_{\Omega}^{-1}M\begin{bmatrix}
P_{X|Y_1}^{-1}J_u\\0\end{bmatrix}\right)(i).
\end{align}
Other elements of $V_{\Omega}^*$ are set to be zero. Now we approximate the entropy of $V_{\Omega}^*$.
\begin{proposition}\label{koonkos}
Let $V_{\Omega_u}^*$ be an extreme point of the set $\mathbb{S}_u$, then we have
\begin{align*}
H(P_{Y|U=u}) &=\sum_{y=1}^{|\mathcal{Y}|}-P_{Y|U=u}(y)\log(P_{Y|U=u}(y))\\&=-(b_u+\frac{\epsilon}{P_{U}(u)} a_uJ_u)+o(\epsilon),
\end{align*}
with $b_u = l_u \left(M_{\Omega_u}^{-1}MP_Y\right),\ 
a_u = l_u\left(M_{\Omega_u}^{-1}M(1\!\!:\!\!|\mathcal{X}|)P_{X|Y_1}^{-1}\right)\in\mathbb{R}^{1\times|\mathcal{X}|},\
l_u = \left[\log\left(M_{\Omega_u}^{-1}MP_{Y}(i)\right)\right]_{i=1:|\mathcal{X}|}\in\mathbb{R}^{1\times|\mathcal{X}|},
$ and $M_{\Omega_u}^{-1}MP_{Y}(i)$ stands for $i$-th ($1\leq i\leq |\mathcal{X}|$) element of the vector $M_{\Omega_u}^{-1}MP_{Y}$. Furthermore, $M(1\!\!:\!\!|\mathcal{X}|)$ stands for submatrix of $M$ with first $|\mathcal{X}|$ columns.
\end{proposition}
\begin{proof}
	The proof follows similar lines as \cite[Lemma~4]{Khodam22} and is based on first order Taylor expansion of $\log(1+x)$.
\end{proof}
By using Proposition \ref{koonkos} we can approximate \eqref{main2} as follows.
\begin{proposition}\label{baghal}
	For sufficiently small $\epsilon$, the minimization problem in \eqref{equii} can be approximated as follows
	\begin{align}\label{kospa}
	&\min_{P_U(.),\{J_u, u\in\mathcal{U}\}} -\left(\sum_{u=1}^{|\mathcal{Y}|} P_U(u)b_u+\epsilon a_uJ_u\right)\\\nonumber
	&\text{subject to:}\\\nonumber
	&\sum_{u=1}^{|\mathcal{Y}|} P_U(u)V_{\Omega_u}^*=P_Y,\ \sum_{u=1}^{|\mathcal{Y}|} J_u=0,\ P_U\in \mathbb{R}_{+}^{|\cal Y|},\\\nonumber
	&\bm{1}^T |J_u|\leq 1,\  \bm{1}^T\cdot J_u=0,\ \forall u\in\mathcal{U},
	\end{align} 
	where $a_u$ and $b_u$ are defined in Proposition~\ref{koonkos}.
\end{proposition}
By using the vector
$\eta_u=P_U(u)\left(M_{\Omega_u}^{-1}MP_Y\right)+\epsilon \left(M_{\Omega_u}^{-1}M(1:|\mathcal{X}|)P_{X|Y_1}^{-1}\right)(J_u)$ for all $u\in \mathcal{U}$, where $\eta_u\in\mathbb{R}^{|\mathcal{X}|}$, we can write \eqref{kospa} as a linear program. The vector $\eta_u$ corresponds to multiple of non-zero elements of the extreme point $V_{\Omega_u}^*$, furthermore, $P_U(u)$ and $J_u$ can be uniquely found as
\begin{align*}
P_U(u)&=\bm{1}^T\cdot \eta_u,\\
J_u&=\frac{P_{X|Y_1}M(1:|\mathcal{X}|)^{-1}M_{\Omega_u}[\eta_u\!-\!(\bm{1}^T \eta_u)M_{\Omega_u}^{-1}MP_Y]}{\epsilon}.
\end{align*}
By solving the linear program we obtain $P_U$ and $J_u$ for all $u$, thus, $P_{Y|U}(\cdot|u)$ can be computed using \eqref{defin1}.
\begin{lemma}\label{lg1}
	Let $P_{U|Y}^*$ be found by the linear program which solves \eqref{kospa} and let $I(U^*;Y)$ be evaluated by this kernel. Then we have
	\begin{align*}
	g_{\epsilon}^1(P_{XY})\geq I(U^*;Y) = L_{g_1}^1(\epsilon).
	\end{align*}
\end{lemma}
\begin{proof}
	The proof follows since the kernel $P_{U|Y}^*$ that achieves the approximate solution satisfies the constraints in \eqref{main2}.
\end{proof}
In the next result we present lower and upper bounds of $g_{\epsilon}^1(P_{XY})$ and $h_{\epsilon}^1(P_{XY})$.
\begin{theorem}\label{choon1}
For sufficiently small $\epsilon\geq 0$ and any pair of RVs $(X,Y)$ distributed according to $P_{XY}$ supported on alphabets $\mathcal{X}$ and $\mathcal{Y}$ we have
\begin{align*}
L_{g_1}^1(\epsilon)\leq g_{\epsilon}^1(P_{XY}),
\end{align*}	
and for any $\epsilon\geq 0$ we obtain
\begin{align*}
 g_{\epsilon}^1(P_{XY})&\leq \frac{\epsilon|\mathcal{Y}||\mathcal{X}|}{\min P_X}+H(Y|X)=U_{g_1}(\epsilon),\\
g_{\epsilon}^1(P_{XY})&\leq h_{\epsilon}^1(P_{XY}).
\end{align*}
Furthermore, for any $0\leq \epsilon\leq \sqrt{2I(X;Y)}$ we have
\begin{align*}
\max\{L_{h_1}^{1}(\epsilon),L_{h_1}^{2}(\epsilon)\}\leq h_{\epsilon}^1(P_{XY}),
\end{align*}
where $L_{h^1}^{1}(\epsilon)$ and $L_{h^1}^{2}(\epsilon)$ are defined in Proposition~\ref{prop111}.
\end{theorem}
\begin{proof}
	Lower bounds on $ g_{\epsilon}^1(P_{XY})$ and $h_{\epsilon}^1(P_{XY})$ are derived in Lemma~\ref{lg1} and Proposition~\ref{prop111}, respectively. Furthermore, inequality $g_{\epsilon}^1(P_{XY})\leq h_{\epsilon}^1(P_{XY})$ holds since $h_{\epsilon}^1(P_{XY})$ has less constraints. To prove the upper bound on $g_{\epsilon}^1(P_{XY})$, i.e., $U_{g_1}(\epsilon)$, let $U$ satisfy $X-Y-U$ and $d(P_{X,U}(\cdot,u),P_XP_U(u))\leq\epsilon$, then we have
	\begin{align*}
	I(U;Y) &= I(X;U)\!+\!H(Y|X)\!-\!I(X;U|Y)\!-\!H(Y|X,U)\\
	&\stackrel{(a)}{=} I(X;U)\!+\!H(Y|X)-H(Y|X,U)\\&\leq I(X;U)\!+\!H(Y|X)\\&=\sum_u P_U(u)D(P_{X|U}(\cdot|u),P_X)+H(Y|X)\\ &\stackrel{(b)}{\leq}\!\sum_u\!\! P_U(u)\frac{\left(d(P_{X|U}(\cdot|u),\!P_X)\right)^2}{\min P_X}\!+\!H(Y|X) \\&\stackrel{(c)}{\leq}\!\sum_u\!\! P_U(u)\frac{d(P_{X|U}(\cdot|u),\!P_X)}{\min P_X}|\mathcal{X}|\!+\!H(Y|X) \\&=\sum_u \!\!\frac{d(P_{X|U}(\cdot,u),\!P_XP_U(u))}{\min P_X}|\mathcal{X}|\!+\!H(Y|X)\\&\stackrel{(d)}{\leq} \frac{\epsilon|\mathcal{Y}||\mathcal{X}|}{\min P_X}+H(Y|X), 
	\end{align*} 
	where (a) follows by the Markov chain $X-Y-U$, (b) follows by the reverse Pinsker inequality \cite[(23)]{verdu} and (c) holds since $d(P_{X|U}(\cdot|u),\!P_X)=\sum_{i=1}^{|\mathcal{X}|} |P_{X|U}(x_i|u)-P_X(x_i)|\leq |\mathcal{X}|$. Latter holds since for each $u$ and $i$, $|P_{X|U}(x_i|u)-P_X|\leq 1$. Moreover, (d) holds since by Proposition~\ref{prop222} without loss of optimality we can assume $|\mathcal{U}|\leq |\mathcal{Y}|$. In other words (d) holds since by Proposition~\ref{prop222} we have
	\begin{align}
	g_{\epsilon}^1(P_{XY})&=\sup_{\begin{array}{c} 
		\substack{P_{U|Y}:X-Y-U\\ \ d(P_{X,U}(\cdot,u),P_XP_{U}(u))\leq\epsilon,\ \forall u}
		\end{array}}I(Y;U)\nonumber\\&= \max_{\begin{array}{c} 
		\substack{P_{U|Y}:X-Y-U\\ \ d(P_{X,U}(\cdot,u),P_XP_{U}(u))\leq\epsilon,\ \forall u\\ |\mathcal{U}|\leq |\mathcal{Y}|}
		\end{array}}I(Y;U).\label{koontala}
	\end{align}
\end{proof}
In the next section we provide bounds for $g_{\epsilon}^2(P_{XY})$ and $h_{\epsilon}^2(P_{XY})$.
\subsection{Lower and Upper bounds on $g_{\epsilon}^2(P_{XY})$ and $h_{\epsilon}^2(P_{XY})$}
As we mentioned earlier in \cite{Khodam22}, we have provided an approximate solution for $g_{\epsilon}^2(P_{XY})$ using local approximation of $H(Y|U)$ for sufficiently small $\epsilon$. Furthermore, in \cite[Proposition~8]{Khodam22} we specified permissible leakages. By using \cite[Proposition~8]{Khodam22}, we can write
\begin{align}
g_{\epsilon}^2(P_{XY})&=\sup_{\begin{array}{c} 
	\substack{P_{U|Y}:X-Y-U\\ \ d(P_{X|U}(\cdot|u),P_X)\leq\epsilon,\ \forall u}
	\end{array}}I(Y;U)\nonumber\\&= \max_{\begin{array}{c} 
	\substack{P_{U|Y}:X-Y-U\\ \ d(P_{X|U}(\cdot|u),P_X)\leq\epsilon,\ \forall u\\ |\mathcal{U}|\leq |\mathcal{Y}|}
	\end{array}}I(Y;U).\label{antar}
\end{align}
In the next lemma we find a lower bound for $g_{\epsilon}^2(P_{XY})$, where we use the approximate problem for \eqref{main22}.
\begin{lemma}\label{lg2}
	Let the kernel $P_{U^*|Y}$ achieve the optimum solution in \cite[Theorem~2]{Khodam22}. Thus, $I(U^*;Y)$ evaluated by this kernel is a lower bound for $g_{\epsilon}^2(P_{XY})$. In other words, we have
		\begin{align*}
	g_{\epsilon}^2(P_{XY})\geq I(U^*;Y) = L_{g_2}^1(\epsilon).
	\end{align*}
\end{lemma}
\begin{proof}
	The proof follows since the kernel $P_{U|Y}^*$ that achieves the approximate solution satisfies the constraints in \eqref{main22}.
\end{proof}
Next we provide upper bounds for $g_{\epsilon}^2(P_{XY})$. To do so, we first bound the approximation error in \cite[Theorem~2]{Khodam22}. Let $\Omega^1$ be the set of all $\Omega_i\subset\{1,..,|\mathcal{Y}|\},\ |\Omega_i|=|\cal X|$, such that each $\Omega_i$ produces a valid standard distribution vector $M_{\Omega_i}^{-1}MP_Y$, i.e., all elements in the vector $M_{\Omega_i}^{-1}MP_Y$ are positive.
\begin{proposition}\label{mos}
	Let the approximation error be the distance between $H(Y|U)$ and the approximation derived in \cite[Theorem~2]{Khodam22}. Then, for all $\epsilon<\frac{1}{2}\epsilon_2$, we have 
	\begin{align*}
	|\text{Approximation\  error}|<\frac{3}{4}.
	\end{align*}
	Furthermore, for all $\epsilon<\frac{1}{2}\frac{\epsilon_2}{\sqrt{|\mathcal{X}|}}$ the upper bound can be strengthened as follows
	\begin{align*}
	|\text{Approximation\  error}|<\frac{1}{2(2\sqrt{|\mathcal{X}|}-1)^2}+\frac{1}{4|\mathcal{X}|}.
	\end{align*}
	where $\epsilon_2=\frac{\min_{y,\Omega\in \Omega^1} M_{\Omega}^{-1}MP_Y(y)}{\max_{\Omega\in \Omega^1} |\sigma_{\max} (H_{\Omega})|}$, $H_{\Omega}=M_{\Omega}^{-1}M(1:|\mathcal{X}|)P_{X|Y_1}^{-1}$ and $\sigma_{\max}$ is the largest right singular value. 
\end{proposition} 
\begin{proof}
	The proof is provided in Appendix~C.
\end{proof}
As a result we can find an upper bound on $g_{\epsilon}^2(P_{XY})$. To do so let $\text{approx}(g_{\epsilon}^2)$ be the value that the Kernel $P_{U^*|Y}$ in Lemma~\ref{lg2} achieves, i.e., the approximate value in \cite[(7)]{Khodam22}.
\begin{corollary}\label{ghahve}
	For any $0\leq\epsilon<\frac{1}{2}\epsilon_2$ we have
	\begin{align*}
	g_{\epsilon}^2(P_{XY})\leq \text{approx}(g_{\epsilon}^2)+\frac{3}{4}=U_{g_2}^1(\epsilon),
	\end{align*}
	furthermore, for any $0\leq\epsilon<\frac{1}{2}\frac{\epsilon_2}{\sqrt{|\mathcal{X}|}}$ the upper bound can be strengthened as
	\begin{align*}
	g_{\epsilon}^2(P_{XY})\!\leq \text{approx}(g_{\epsilon}^2)+\frac{1}{2(2\sqrt{|\mathcal{X}|}-1)^2}\!+\frac{1}{4|\mathcal{X}|}\!=\!U_{g_2}^2(\epsilon).
	\end{align*}
\end{corollary}
In the next theorem we summarize the bounds for $g_{\epsilon}^2(P_{XY})$ and $h_{\epsilon}^2(P_{XY})$, furthermore, a new upper bound for $h_{\epsilon}^2(P_{XY})$ is derived.
\begin{theorem}\label{koontala1}
	For any $0\leq\epsilon<\frac{1}{2}\epsilon_2$ and pair of RVs $(X,Y)$ distributed according to $P_{XY}$ supported on alphabets $\mathcal{X}$ and $\mathcal{Y}$ we have
	\begin{align*}
		L_{g_2}^1(\epsilon)\leq g_{\epsilon}^2(P_{XY})\leq U_{g_2}^1(\epsilon),
	\end{align*}
	and for any $0\leq\epsilon<\frac{1}{2}\frac{\epsilon_2}{\sqrt{|\mathcal{X}|}}$ we get
	\begin{align*}
	L_{g_2}^1(\epsilon)\leq g_{\epsilon}^2(P_{XY})\leq U_{g_2}^2(\epsilon),
	\end{align*}
	furthermore, for any $0\leq\epsilon$ 
	\begin{align*}
	g_{\epsilon}^2(P_{XY})\leq h_{\epsilon}^2(P_{XY})\leq \frac{\epsilon^2}{\min P_X}+H(Y|X)=U_{h_2}(\epsilon).
	\end{align*}
\end{theorem}
\begin{proof}
	It is sufficient to show that the upper bound on $h_{\epsilon}^2(P_{XY})$ holds, i.e., $U_{h_2}(\epsilon)$. To do so, let $U$ satisfy $d(P_{X|U}(\cdot|u),P_X)\leq\epsilon$, then we have
	\begin{align*}
	I(U;Y) &= I(X;U)\!+\!H(Y|X)\!-\!I(X;U|Y)\!-\!H(Y|X,U)\\ &\leq I(X;U)\!+\!H(Y|X)\\ &\stackrel{(a)}{\leq}\!\sum_u\!\! P_U(u)\frac{\left(d(P_{X|U}(\cdot|u),\!P_X)\right)^2}{\min P_X}\!+\!H(Y|X)  \\& = \frac{\epsilon^2}{\min P_X}+H(Y|X),
	\end{align*}
	where (a) follows by the reverse Pinsker inequality.
\end{proof}
In next section we study the special case where $X$ is a deterministic function of $Y$, i.e., $H(X|Y)=0$. 
\subsection{Special case: $X$ is a deterministic function of $Y$}
In this case we have 
\begin{align}
h_{\epsilon}^1(P_{XY})&=g_{\epsilon}^1(P_{XY})\label{khatakar}\\&= \max_{\begin{array}{c} 
	\substack{P_{U|Y}:X-Y-U\\ \ d(P_{X,U}(\cdot,u),P_XP_{U}(u))\leq\epsilon,\ \forall u\\ |\mathcal{U}|\leq |\mathcal{Y}|}
	\end{array}}I(Y;U)\nonumber\\&= \sup_{\begin{array}{c} 
	\substack{P_{U|Y}: d(P_{X,U}(\cdot,u),P_XP_{U}(u))\leq\epsilon,\ \forall u\\ |\mathcal{U}|\leq |\mathcal{Y}|}
	\end{array}}I(Y;U),\nonumber\\
h_{\epsilon}^2(P_{XY})&=g_{\epsilon}^2(P_{XY})\label{khatakoon}\\&= \max_{\begin{array}{c} 
	\substack{P_{U|Y}:X-Y-U\\ \ d(P_{X|U}(\cdot|u),P_X)\leq\epsilon,\ \forall u\\ |\mathcal{U}|\leq |\mathcal{Y}|}
	\end{array}}I(Y;U)\nonumber\\&= \sup_{\begin{array}{c} 
	\substack{P_{U|Y}: d(P_{X|U}(\cdot|u),P_X)\leq\epsilon,\ \forall u\\ |\mathcal{U}|\leq |\mathcal{Y}|}
	\end{array}}I(Y;U),\nonumber
\end{align}
since the Markov chain $X-Y-U$ holds. Consequently, by using Theorem~2 and \eqref{khatakar} we have next corollary.
\begin{corollary}\label{chooni}
	For any $0\leq \epsilon\leq \sqrt{2I(X;Y)}$ we have
	\begin{align*}
\max\{L_{h_1}^{1}(\epsilon),L_{h_1}^{2}(\epsilon),L_{g_1}^1(\epsilon)\}\leq g_{\epsilon}^1(P_{XY})\leq U_{g_1}(\epsilon).	
	\end{align*}
\end{corollary}
We can see that the bounds in Corollary~\ref{chooni} are asymptotically optimal. The latter follows since in high privacy regimes, i.e., the leakage tends to zero, $U_{g_1}(\epsilon)$ and $L_{h_1}^{1}(\epsilon)$ both tend to $H(Y|X)$, which is the optimal solution to $g_{0}(P_{XY})$ when $X$ is a deterministic function of $Y$, \cite[Theorem~6]{kostala}.
Furthermore, by using Theorem~3 and \eqref{khatakoon} we obtain the next result.
\begin{corollary}
For any $0\leq\epsilon<\frac{1}{2}\epsilon_2$ we have
\begin{align*}
L_{g_2}^1(\epsilon)\leq g_{\epsilon}^2(P_{XY}) \leq \min\{U_{g_2}^1(\epsilon),U_{h_2}(\epsilon)\}.
\end{align*}	
\end{corollary}
\begin{remark}
	For deriving the upper bound $U_{h_2}(\epsilon)$ and lower bounds $L_{h_1}^1(\epsilon)$ and $L_{h_1}^2(\epsilon)$ we do not use the assumption that the leakage matrix $P_{X|Y}$ is of full row rank. Thus, these bounds hold for all $P_{X|Y}$ and all $\epsilon\geq 0$.  
\end{remark}
In the next part, we study a numerical example to illustrate the new bounds.
\subsection{Example}
Let us consider RVs $X$ and $Y$ with joint distribution $P_{XY}=\begin{bmatrix}
0.693 & 0.027 &0.108& 0.072\\0.006 & 0.085 & 0.004 & 0.005
\end{bmatrix}$. Using definition of $\epsilon_2$ in Proposition~\ref{mos} we have $\epsilon_2 = 0.0341$. Fig.~\ref{kir12} illustrates the lower bound and upper bounds for $g_{\epsilon}^2$ derived in Theorem~\ref{koontala1}. As shown in Fig.~\ref{kir12}, the upper bounds $U_{g_2}^1(\epsilon)$ and $U_{g_2}^2(\epsilon)$ are valid for $\epsilon< 0.0171 $ and $\epsilon < 0.0121$, however the upper bound $U_{h_2}(\epsilon)$ is valid for all $\epsilon\geq 0$. In this example, we can see that for any $\epsilon$ the upper bound $U_{h_2}(\epsilon)$ is the smallest upper bound.
\begin{figure}[]
	\centering
	\includegraphics[scale = .18]{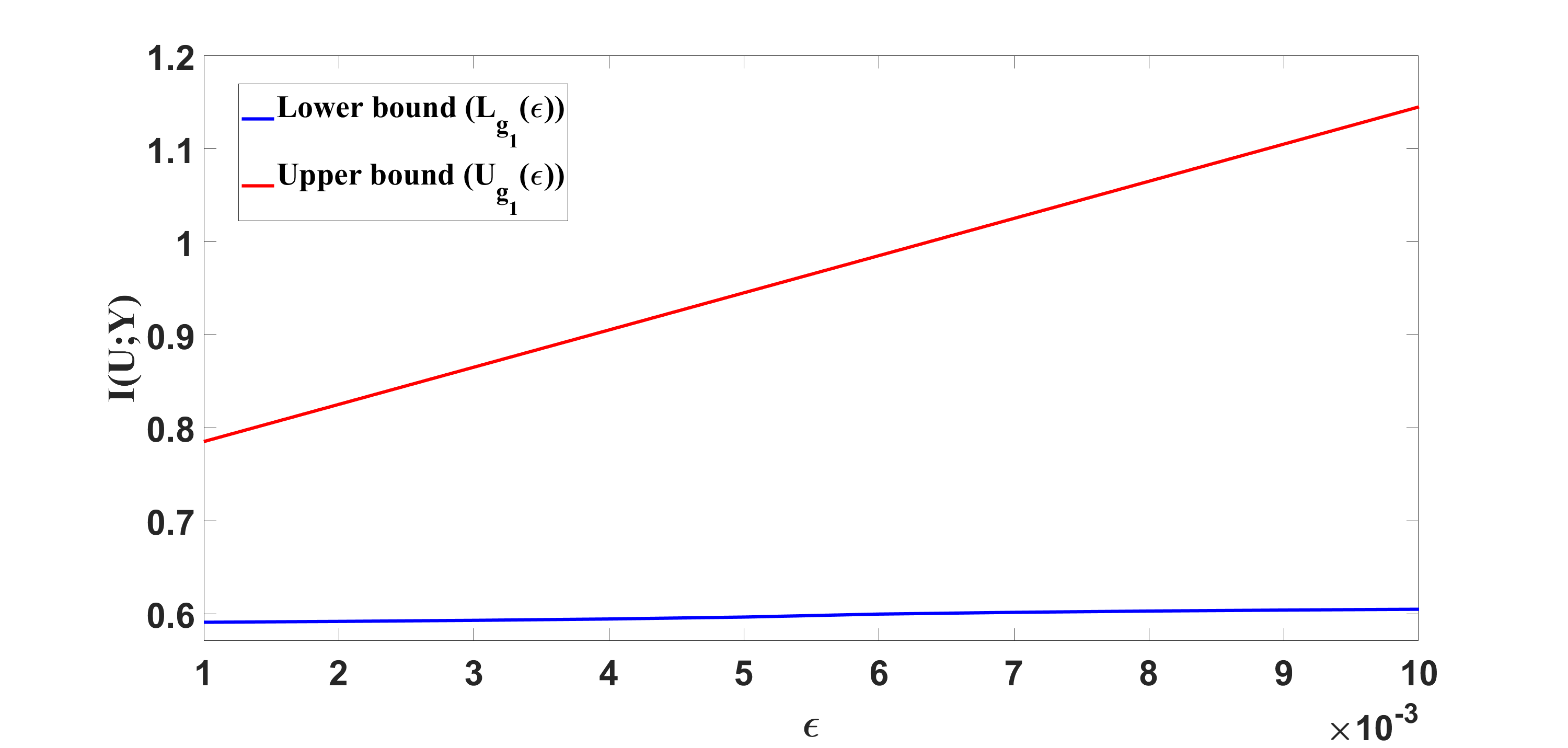}
	\caption{Comparing the upper bound and lower bound for $g_{\epsilon}^1$.}
	\label{kir11}
\end{figure} 
Furthermore, Fig.~\ref{kir11} shows the lower bound $L_{g_1}(\epsilon)$ and upper bound $U_{g_1}(\epsilon)$ obtained in Theorem~\ref{choon1}. 
\begin{figure}[]
	\centering
	\includegraphics[scale = .18]{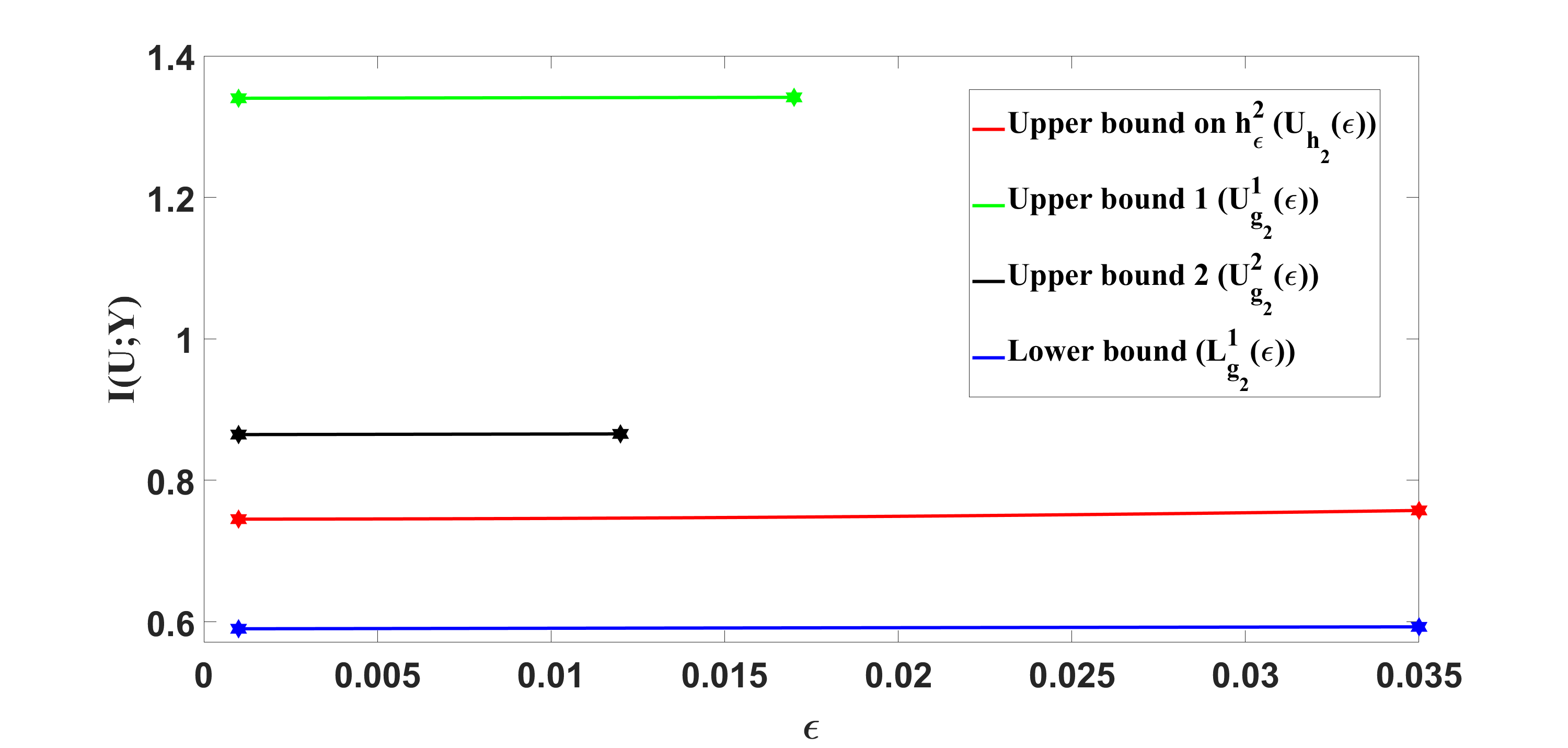}
	\caption{Comparing the upper bound and lower bound for $g_{\epsilon}^2$. The upper bounds $U_{g_2}^1(\epsilon)$ and $U_{g_2}^2(\epsilon)$ are valid for $\epsilon< 0.0171 $ and $\epsilon < 0.0121$, respectively. On the other hand, the upper bound $U_{h_2}(\epsilon)$ is valid for all $\epsilon\geq 0$.}
	\label{kir12}
\end{figure} 
\begin{figure}[]
	\centering
	\includegraphics[scale = .18]{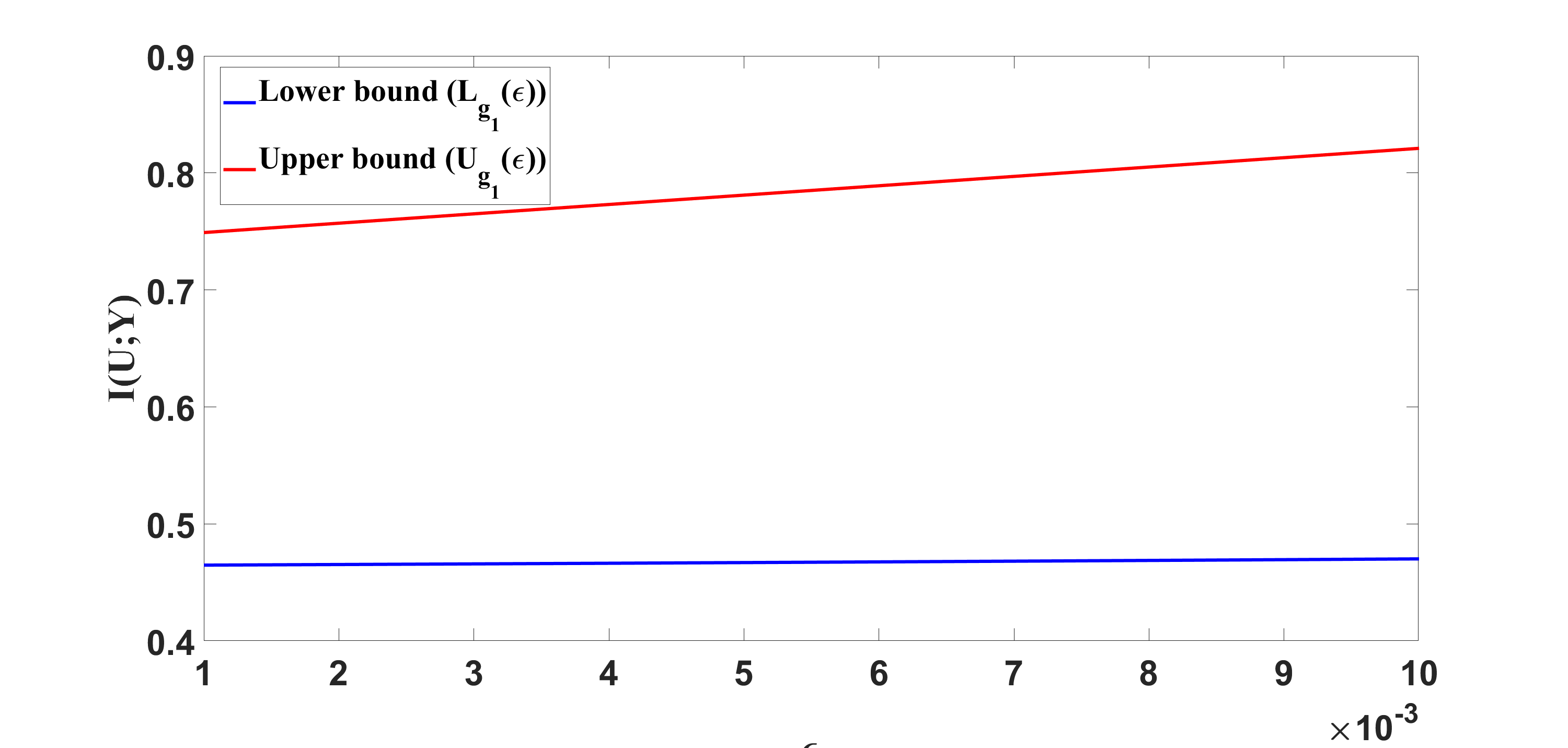}
	\caption{Comparing the upper bound and lower bound for $g_{\epsilon}^1$.}
	\label{kir111}
\end{figure}  
\begin{figure}[]
	\centering
	\includegraphics[scale = .18]{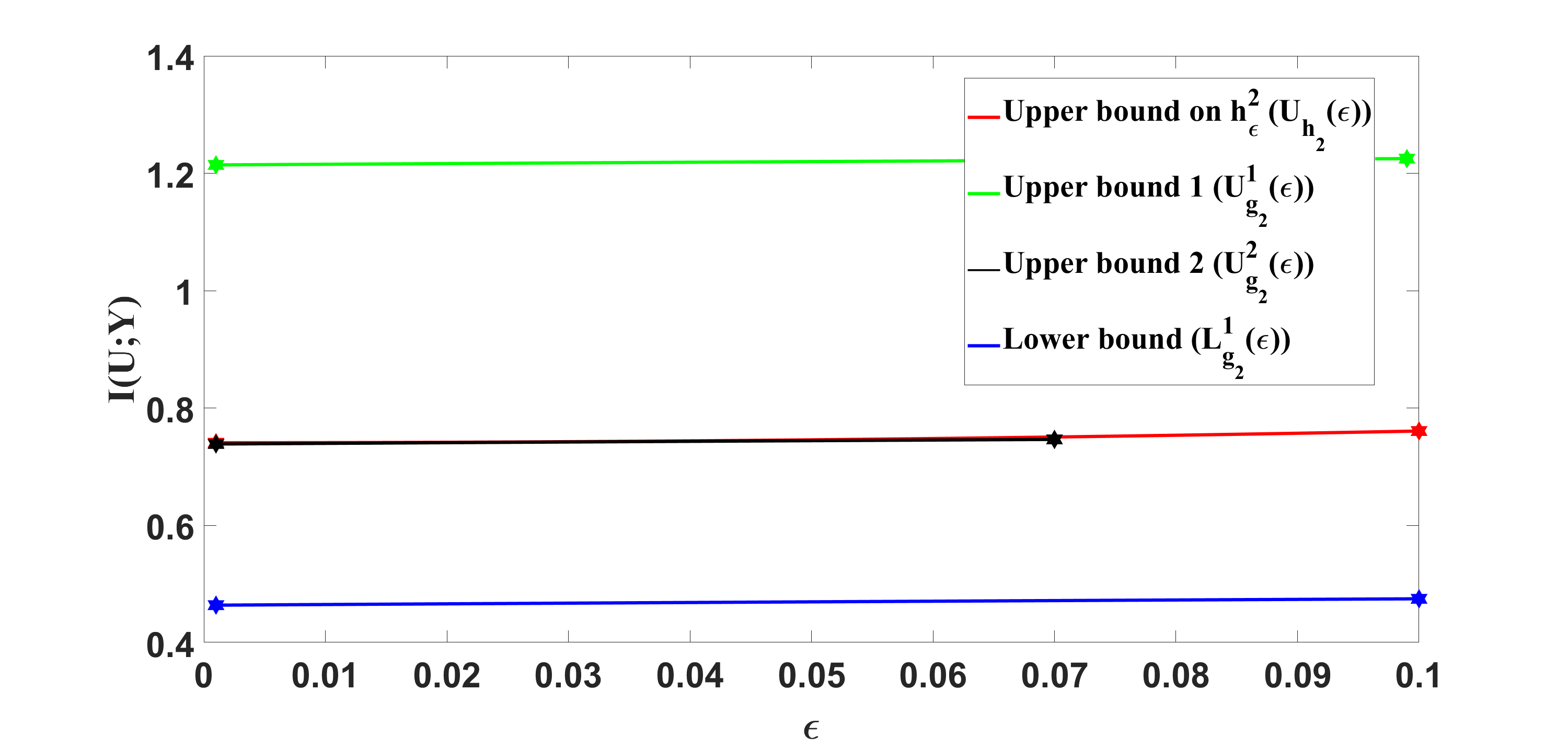}
	\caption{Comparing the upper bound and lower bound for $g_{\epsilon}^2$. The upper bounds $U_{g_2}^1(\epsilon)$ and $U_{g_2}^2(\epsilon)$ are valid for $\epsilon<0.0997 $ and $\epsilon <0.0705$, respectively. However, the upper bound $U_{h_2}(\epsilon)$ is valid for all $\epsilon\geq 0$.}
	\label{kir122}
\end{figure}  
Next, let $P_{XY}=\begin{bmatrix}
0.350 & 0.025 &0.085& 0.040\\0.025 & 0.425 & 0.035 & 0.015
\end{bmatrix}$. In this case, $\epsilon_2=0.1994$. Fig.~\ref{kir122} illustrates the lower bound and upper bounds for $g_{\epsilon}^2$. We can see that for $\epsilon<0.0705$, $U_{g_2}^2(\epsilon)$ is the smallest upper bound and for $\epsilon>0.0705$, $U_{h_2}(\epsilon)$ is the smallest bound. Furthermore, Fig.~\ref{kir111} shows the lower bound $L_{g_1}(\epsilon)$ and upper bound $U_{g_1}(\epsilon)$.
\section*{Appendix A}
For each $u\in\mathcal{U}$ we have
\begin{align*}
\mathcal{L}^1(X;U=u)&= \left\lVert P_{X|U}(\cdot|u)-P_X \right\rVert_1\\&=\left\lVert P_{X|Y}(P_{Y|U}(\cdot|u)-P_Y)\right\rVert_1\\&=\sum_x |\sum_y P_{X|Y}(x,y)(P_{Y|u}(y)\!-\!P_Y(y))|\\
& \stackrel{(a)}{\leq} \sum_x\sum_y P_{X|Y}(x,y)|P_{Y|u}(y)-P_Y(y)|\\&=\sum_y\sum_x P_{X|Y}(x,y)|P_{Y|u}(y)-P_Y(y)|\\&=\sum_y |P_{Y|u}(y)-P_Y(y)|\\&= \left\lVert P_{Y|U=u}\!-\!P_Y \right\rVert_1=\mathcal{L}^1(Y;U=u),
\end{align*}
where (a) follows from the triangle inequality. Furthermore, we can multiply all the above expressions by the term $P_U(u)$ and we obtain
\begin{align*}
\mathcal{L}^2(X;U=u)\leq \mathcal{L}^2(Y;U=u).
\end{align*} 
\section*{Appendix B}
Let $U$ be found by EFRL as in \cite[Lemma~3]{zamani2022bounds}, where we let the leakage be $\frac{\epsilon^2}{2}$. Thus, we have
\begin{align*}
\frac{\epsilon^2}{2}=I(U;X)&=\sum_u P_U(u)D(P_{X|U}(\cdot|u),P_X)\\ &\stackrel{(a)}{\geq} \sum_u \frac{P_U(u)}{2}\left( d(P_{X|U}(\cdot|u),P_X)\right)^2\\&\stackrel{(b)}{\geq} \sum_u \frac{P_U(u)^2}{2}\left( d(P_{X|U}(\cdot|u),P_X)\right)^2\\ &\geq \frac{P_U(u)^2}{2}\left( d(P_{X|U}(\cdot|u),P_X)\right)^2\\ &= \frac{\left( d(P_{X,U}(\cdot,u),P_XP_U(u))\right)^2}{2},
\end{align*}
where $D(\cdot,\cdot)$ corresponds to KL-divergence. Furthermore, (a) follows by the Pinsker’s inequality \cite{verdu} and (b) follows since $0\leq P_U(u)\leq 1$. Using the last line we obtain
\begin{align*}
d(P_{X,U}(\cdot,u),P_XP_U(u))\leq \epsilon,\ \forall u.
\end{align*}
The other constraints can be obtained by using \cite[Lemma~3]{zamani2022bounds}.
\section*{Appendix C}
By using \cite[Proposition~2]{Khodam22}, it suffices to assume $|\mathcal{U}|\leq|\mathcal{Y}|$. Using \cite[Proposition~3]{Khodam22}, let us consider $|\mathcal{Y}|$ extreme points that achieves the minimum in \cite[Theorem~2]{Khodam22} as $V_{\Omega_j}$ for $j\in\{1,..,|\mathcal{Y}|\}$. 
Let $|\mathcal{X}|$ non-zero elements of $V_{\Omega_j}$ be $a_{ij}+\epsilon b_{ij}$ for  $i\in\{1,..,|\mathcal{X}|\}$ and $j\in\{1,..,|\mathcal{Y}|\}$, where $a_{ij}$ and $b_{ij}$ can be found in \cite[(6)]{Khodam22}. 
As a summary for $i\in\{1,..,|\mathcal{X}|\}$ and $j\in\{1,..,|\mathcal{Y}|\}$ we have $\sum_i a_{ij}=1$, $\sum_i b_{ij}=0$, $0\leq a_{ij}\leq 1$, and $0\leq a_{ij}+\epsilon b_{ij}\leq1.$
We obtain
\begin{align*}
\max I(U;Y)&=H(Y)\\&+\sum_jP_j\sum_i (a_{ij}+\epsilon b_{ij})\log(a_{ij}+\epsilon b_{ij}),\\
&=H(Y)+\sum_jP_j\times\\&\sum_i (a_{ij}+\epsilon b_{ij})(\log(a_{ij})+\log(1+\epsilon\frac{b_{ij}}{a_{ij}})).
\end{align*} 
In \cite[Theorem~2]{Khodam22}, we have used the Taylor expansion to derive the approximation of the equivalent problem. From the Taylor's expansion formula we have 
\begin{align*}
f(x)&=f(a)+\frac{f'(a)}{1!}(x-a)+\frac{f''(a)}{2!}(x-a)^2+...\\&+\frac{f^{(n)}(a)}{n!}(x-a)^n+R_{n+1}(x),
\end{align*} 
where
\begin{align}\label{kosss}
R_{n+1}(x)&=\int_{a}^{x}\frac{(x-t)^n}{n!}f^{(n+1)}(t)dt\\&=\frac{f^{(n+1)}(\zeta)}{(n+1)!}(x-a)^{n+1},
\end{align}
for some $\zeta\in[a,x]$. 
In \cite{Khodam22} we approximated the terms $\log(1+\frac{b_{ij}}{a_{ij}}\epsilon)$ by $\frac{b_{ij}}{a_{ij}}\epsilon+o(\epsilon)$. Using \eqref{kosss}, there exists an $\zeta_{ij}\in[0,\epsilon]$ such that the error of approximating the term $\log(1+\epsilon\frac{a_{ij}}{b_{ij}})$ is as follows
\begin{align*}
R_2^{ij}(\epsilon)=-\frac{1}{2}\left(\frac{\frac{b_{ij}}{a_{ij}}}{1+\frac{b_{ij}}{a_{ij}}\zeta_{ij}}\right)^2\epsilon^2=-\frac{1}{2}\left( \frac{b_{ij}}{a_{ij}+b_{ij}\zeta_{ij}}\right)^2\epsilon^2.
\end{align*}
Thus, the error of approximation is as follows
\begin{align}\label{chos}
&\text{Approximation\  error}\\&=\sum_{ij} P_j(a_{ij}+\epsilon b_{ij})R_2^{ij}(\epsilon)+\sum_{ij}P_j\frac{b_{ij}^2}{a_{ij}}\epsilon^2\nonumber
\\ &= -\sum_{ij} P_j(a_{ij}+\epsilon b_{ij})\frac{1}{2}\left( \frac{b_{ij}}{a_{ij}+b_{ij}\zeta_{ij}}\right)^2\!\!\epsilon^2\!+\!\sum_{ij}P_j\frac{b_{ij}^2}{a_{ij}}\epsilon^2
\end{align}
An upper bound on approximation error can be obtained as follows
\begin{align}\label{antala}
&|\text{Approximation\  error}|\\&\leq |\sum_{ij} P_j(a_{ij}+\epsilon b_{ij})\frac{1}{2}\left( \frac{b_{ij}}{a_{ij}+b_{ij}\zeta_{ij}}\right)^2\epsilon^2|\\&+|\sum_{ij}P_j\frac{b_{ij}^2}{a_{ij}}\epsilon^2|.
\end{align}
By using the definition of $\epsilon_2$ in Proposition~5 we have $\epsilon<\epsilon_2$ implies $\epsilon<\frac{\min_{ij} a_{ij}}{\max_{ij} |b_{ij}|}$, since $\min_{ij} a_{ij}=\min_{y,\Omega\in \Omega^1} M_{\Omega}^{-1}MP_Y(y)$ and $\max_{ij} |b_{ij}|<\max_{\Omega\in \Omega^1} |\sigma_{\max} (H_{\Omega})|$. By using the upper bound $\epsilon<\frac{\min_{ij} a_{ij}}{|\max_{ij} b_{ij}|}$ we can bound the second term in \eqref{antala} by $1$, since we have 
\begin{align*}
|\sum_{ij}P_j\frac{b_{ij}^2}{a_{ij}}\epsilon^2|&<|\sum_{ij}P_j \frac{b_{ij}^2}{a_{ij}}\left(\frac{\min_{ij} a_{ij}}{\max_{ij} |b_{ij}|}\right)^2|\\&<|\sum_{ij} P_j\min_{ij} a_{ij}|=|\mathcal{X}|\min_{ij} a_{ij}\stackrel{(a)}{<} 1,
\end{align*}
where (a) follows from $\sum_{i} a_{ij}=1,\ \forall j\in\{1,..,|\mathcal{Y}|\}$.
\\If we use $\frac{1}{2}\epsilon_2$ as an upper bound on $\epsilon$, we have $\epsilon<\frac{1}{2}\frac{\min_{ij} a_{ij}}{\max_{ij} |b_{ij}|}$. We show that by using this upper bound the first term in \eqref{antala} can be upper bounded by $\frac{1}{2}$. We have
\begin{align*}
&\frac{1}{2}|\sum_{ij} P_j(a_{ij}+\epsilon b_{ij})\left( \frac{b_{ij}}{a_{ij}+b_{ij}\zeta_{ij}}\right)^2\epsilon^2|\\&\stackrel{(a)}{<}\frac{1}{2}|\sum_{ij} P_j(a_{ij}+\epsilon b_{ij})\left(\frac{|b_{ij}|}{a_{ij}-\epsilon|b_{ij}|}\epsilon\right)^2|\\&\stackrel{(b)}{<}\frac{1}{2}|\sum_{ij} P_j(a_{ij}+\epsilon b_{ij})|<\frac{1}{2},
\end{align*}  
where (a) follows from $0\leq\zeta_{ij}\leq \epsilon,\ \forall i,\ \forall j,$ and (b) follows from $\frac{|b_{ij}|}{a_{ij}-\epsilon|b_{ij}|}\epsilon<1$ for all $i$ and $j$. The latter can be shown as follows
\begin{align*}
\frac{|b_{ij}|}{a_{ij}-\epsilon|b_{ij}|}\epsilon<\frac{|b_{ij}|}{a_{ij}-\frac{1}{2}\frac{\min_{ij} a_{ij}}{\max_{ij} |b_{ij}|}|b_{ij}|}\epsilon<\frac{b_{ij}}{\frac{1}{2}\min_{ij}a_{ij}}\epsilon<1.
\end{align*}
For $\epsilon<\frac{1}{2}\epsilon_2$ the term $a_{ij}-\epsilon|b_{ij}|$ is positive and there is no need of absolute value for this term. 
Thus, $\epsilon<\frac{1}{2}\epsilon_2$ implies the following upper bound
\begin{align*}
|\text{Approximation\  error}|<\frac{3}{4}.
\end{align*}
Furthermore, by following similar steps if we use the upper bound $\epsilon<\frac{1}{2}\frac{\epsilon_2}{\sqrt{|\mathcal{X}|}}$ instead of $\epsilon<\frac{1}{2}\epsilon_2$, the upper bound on error can be strengthened by
\begin{align*}
|\text{Approximation\  error}|<\frac{1}{2(2\sqrt{|\mathcal{X}|}-1)^2}+\frac{1}{4|\mathcal{X}|}.
\end{align*}
\bibliographystyle{IEEEtran}
\bibliography{IEEEabrv,IZS}

\begin{thebibliography}{10}
\providecommand{\url}[1]{#1}
\csname url@samestyle\endcsname
\providecommand{\newblock}{\relax}
\providecommand{\bibinfo}[2]{#2}
\providecommand{\BIBentrySTDinterwordspacing}{\spaceskip=0pt\relax}
\providecommand{\BIBentryALTinterwordstretchfactor}{4}
\providecommand{\BIBentryALTinterwordspacing}{\spaceskip=\fontdimen2\font plus
\BIBentryALTinterwordstretchfactor\fontdimen3\font minus
  \fontdimen4\font\relax}
\providecommand{\BIBforeignlanguage}[2]{{%
\expandafter\ifx\csname l@#1\endcsname\relax
\typeout{** WARNING: IEEEtran.bst: No hyphenation pattern has been}%
\typeout{** loaded for the language `#1'. Using the pattern for}%
\typeout{** the default language instead.}%
\else
\language=\csname l@#1\endcsname
\fi
#2}}
\providecommand{\BIBdecl}{\relax}
\BIBdecl

\bibitem{Calmon2}
H.~{Wang}, L.~{Vo}, F.~P. {Calmon}, M.~{M\'{e}dard}, K.~R. {Duffy}, and
  M.~{Varia}, ``Privacy with estimation guarantees,'' \emph{IEEE Transactions
  on Information Theory}, vol.~65, no.~12, pp. 8025--8042, Dec 2019.

\bibitem{yamamoto}
H.~Yamamoto, ``A source coding problem for sources with additional outputs to
  keep secret from the receiver or wiretappers (corresp.),'' \emph{IEEE
  Transactions on Information Theory}, vol.~29, no.~6, pp. 918--923, 1983.

\bibitem{sankar}
L.~Sankar, S.~R. Rajagopalan, and H.~V. Poor, ``Utility-privacy tradeoffs in
  databases: An information-theoretic approach,'' \emph{IEEE Transactions on
  Information Forensics and Security}, vol.~8, no.~6, pp. 838--852, 2013.

\bibitem{borz}
B.~{Rassouli} and D.~{G\"{u}nd\"{u}z}, ``On perfect privacy,'' \emph{IEEE
  Journal on Selected Areas in Information Theory}, vol.~2, no.~1, pp.
  177--191, 2021.

\bibitem{gun}
S.~{Sreekumar} and D.~{G\"{u}nd\"{u}z}, ``Optimal privacy-utility trade-off
  under a rate constraint,'' in \emph{2019 IEEE International Symposium on
  Information Theory}, July 2019, pp. 2159--2163.

\bibitem{khodam}
A.~Zamani, T.~J. Oechtering, and M.~Skoglund, ``A design framework for strongly
  $\chi^2$-private data disclosure,'' \emph{IEEE Transactions on Information
  Forensics and Security}, vol.~16, pp. 2312--2325, 2021.

\bibitem{Khodam22}
{A. Zamani, T. J. Oechtering, and M. Skoglund}, ``Data disclosure with non-zero
  leakage and non-invertible leakage matrix,'' \emph{IEEE Transactions on
  Information Forensics and Security}, vol.~17, pp. 165--179, 2022.

\bibitem{kostala}
Y.~Y. Shkel, R.~S. Blum, and H.~V. Poor, ``Secrecy by design with applications
  to privacy and compression,'' \emph{IEEE Transactions on Information Theory},
  vol.~67, no.~2, pp. 824--843, 2021.

\bibitem{issa}
I.~{Issa}, S.~{Kamath}, and A.~B. {Wagner}, ``An operational measure of
  information leakage,'' in \emph{2016 Annual Conference on Information Science
  and Systems}, March 2016, pp. 234--239.

\bibitem{makhdoumi}
A.~Makhdoumi, S.~Salamatian, N.~Fawaz, and M.~M{\'e}dard, ``From the
  information bottleneck to the privacy funnel,'' in \emph{2014 IEEE
  Information Theory Workshop}, 2014, pp. 501--505.

\bibitem{dwork1}
C.~Dwork, F.~McSherry, K.~Nissim, and A.~Smith, ``Calibrating noise to
  sensitivity in private data analysis,'' in \emph{Theory of cryptography
  conference}.\hskip 1em plus 0.5em minus 0.4em\relax Springer, 2006, pp.
  265--284.

\bibitem{calmon4}
F.~P. {Calmon}, A.~{Makhdoumi}, M.~{Medard}, M.~{Varia}, M.~{Christiansen}, and
  K.~R. {Duffy}, ``Principal inertia components and applications,'' \emph{IEEE
  Transactions on Information Theory}, vol.~63, no.~8, pp. 5011--5038, Aug
  2017.

\bibitem{issajoon}
I.~Issa, A.~B. Wagner, and S.~Kamath, ``An operational approach to information
  leakage,'' \emph{IEEE Transactions on Information Theory}, vol.~66, no.~3,
  pp. 1625--1657, 2020.

\bibitem{asoo}
S.~Asoodeh, M.~Diaz, F.~Alajaji, and T.~Linder, ``Estimation efficiency under
  privacy constraints,'' \emph{IEEE Transactions on Information Theory},
  vol.~65, no.~3, pp. 1512--1534, 2019.

\bibitem{Total}
B.~Rassouli and D.~{G\"{u}nd\"{u}z}, ``Optimal utility-privacy trade-off with
  total variation distance as a privacy measure,'' \emph{IEEE Transactions on
  Information Forensics and Security}, vol.~15, pp. 594--603, 2020.

\bibitem{issa2}
I.~Issa, S.~Kamath, and A.~B. Wagner, ``Maximal leakage minimization for the
  shannon cipher system,'' in \emph{2016 IEEE International Symposium on
  Information Theory}, 2016, pp. 520--524.

\bibitem{kosnane}
C.~T. Li and A.~E. Gamal, ``Strong functional representation lemma and
  applications to coding theorems,'' \emph{IEEE Transactions on Information
  Theory}, vol.~64, no.~11, pp. 6967--6978, 2018.

\bibitem{shahab}
\BIBentryALTinterwordspacing
S.~Asoodeh, M.~Diaz, F.~Alajaji, and T.~Linder, ``Information extraction under
  privacy constraints,'' \emph{Information}, vol.~7, no.~1, 2016. [Online].
  Available: \url{https://www.mdpi.com/2078-2489/7/1/15}
\BIBentrySTDinterwordspacing

\end{thebibliography}


\begin{thebibliography}{10}
\providecommand{\url}[1]{#1}
\csname url@samestyle\endcsname
\providecommand{\newblock}{\relax}
\providecommand{\bibinfo}[2]{#2}
\providecommand{\BIBentrySTDinterwordspacing}{\spaceskip=0pt\relax}
\providecommand{\BIBentryALTinterwordstretchfactor}{4}
\providecommand{\BIBentryALTinterwordspacing}{\spaceskip=\fontdimen2\font plus
\BIBentryALTinterwordstretchfactor\fontdimen3\font minus
  \fontdimen4\font\relax}
\providecommand{\BIBforeignlanguage}[2]{{%
\expandafter\ifx\csname l@#1\endcsname\relax
\typeout{** WARNING: IEEEtran.bst: No hyphenation pattern has been}%
\typeout{** loaded for the language `#1'. Using the pattern for}%
\typeout{** the default language instead.}%
\else
\language=\csname l@#1\endcsname
\fi
#2}}
\providecommand{\BIBdecl}{\relax}
\BIBdecl

\bibitem{issa}
I.~{Issa}, S.~{Kamath}, and A.~B. {Wagner}, ``An operational measure of
  information leakage,'' in \emph{2016 Annual Conference on Information Science
  and Systems}, March 2016, pp. 234--239.

\bibitem{makhdoumi}
A.~Makhdoumi, S.~Salamatian, N.~Fawaz, and M.~M{\'e}dard, ``From the
  information bottleneck to the privacy funnel,'' in \emph{2014 IEEE
  Information Theory Workshop}, 2014, pp. 501--505.

\bibitem{dwork1}
C.~Dwork, F.~McSherry, K.~Nissim, and A.~Smith, ``Calibrating noise to
  sensitivity in private data analysis,'' in \emph{Theory of cryptography
  conference}.\hskip 1em plus 0.5em minus 0.4em\relax Springer, 2006, pp.
  265--284.

\bibitem{Calmon2}
H.~{Wang}, L.~{Vo}, F.~P. {Calmon}, M.~{M\'{e}dard}, K.~R. {Duffy}, and
  M.~{Varia}, ``Privacy with estimation guarantees,'' \emph{IEEE Transactions
  on Information Theory}, vol.~65, no.~12, pp. 8025--8042, Dec 2019.

\bibitem{yamamoto}
H.~Yamamoto, ``A source coding problem for sources with additional outputs to
  keep secret from the receiver or wiretappers (corresp.),'' \emph{IEEE
  Transactions on Information Theory}, vol.~29, no.~6, pp. 918--923, 1983.

\bibitem{sankar}
L.~Sankar, S.~R. Rajagopalan, and H.~V. Poor, ``Utility-privacy tradeoffs in
  databases: An information-theoretic approach,'' \emph{IEEE Transactions on
  Information Forensics and Security}, vol.~8, no.~6, pp. 838--852, 2013.

\bibitem{deniz6}
B.~{Rassouli} and D.~{G\"{u}nd\"{u}z}, ``On perfect privacy,'' in \emph{2018
  IEEE International Symposium on Information Theory (ISIT)}, June 2018, pp.
  2551--2555.

\bibitem{Amir}
A.~Zamani, T.~J. Oechtering, and M.~Skoglund, ``A design framework for
  epsilon-private data disclosure,'' \emph{arXiv preprint arXiv:2009.01704}, 3
  Sep 2020.

\bibitem{jende}
B.~Razeghi, F.~P. Calmon, D.~{G\"{u}nd\"{u}z}, and S.~Voloshynovskiy, ``On
  perfect obfuscation: Local information geometry analysis,'' in \emph{2020
  IEEE International Workshop on Information Forensics and Security}, 2020, pp.
  1--6.

\bibitem{borade}
S.~Borade and L.~Zheng, ``Euclidean information theory,'' in \emph{2008 IEEE
  International Zurich Seminar on Communications}.\hskip 1em plus 0.5em minus
  0.4em\relax IEEE, 2008, pp. 14--17.

\bibitem{huang}
S.~L. Huang and L.~Zheng, ``Linear information coupling problems,'' in
  \emph{2012 IEEE International Symposium on Information Theory
  Proceedings}.\hskip 1em plus 0.5em minus 0.4em\relax IEEE, 2012, pp.
  1029--1033.

\bibitem{el2011network}
A.~El~Gamal and Y.-H. Kim, \emph{Network information theory}.\hskip 1em plus
  0.5em minus 0.4em\relax Cambridge university press, 2011.

\bibitem{Amir2}
A.~Zamani, T.~J. Oechtering, and M.~Skoglund, ``Data disclosure mechanism
  design with non-zero leakage,'' 2020,
  \url{https://people.kth.se/~oech/ITWlong.pdf}.

\end{thebibliography}


\begin{thebibliography}{10}
\providecommand{\url}[1]{#1}
\csname url@samestyle\endcsname
\providecommand{\newblock}{\relax}
\providecommand{\bibinfo}[2]{#2}
\providecommand{\BIBentrySTDinterwordspacing}{\spaceskip=0pt\relax}
\providecommand{\BIBentryALTinterwordstretchfactor}{4}
\providecommand{\BIBentryALTinterwordspacing}{\spaceskip=\fontdimen2\font plus
\BIBentryALTinterwordstretchfactor\fontdimen3\font minus
  \fontdimen4\font\relax}
\providecommand{\BIBforeignlanguage}[2]{{%
\expandafter\ifx\csname l@#1\endcsname\relax
\typeout{** WARNING: IEEEtran.bst: No hyphenation pattern has been}%
\typeout{** loaded for the language `#1'. Using the pattern for}%
\typeout{** the default language instead.}%
\else
\language=\csname l@#1\endcsname
\fi
#2}}
\providecommand{\BIBdecl}{\relax}
\BIBdecl

\bibitem{rassoul1}
B.~Rassouli and D.~G\"{u}nd\"{u}z, ``On perfect privacy and maximal
  correlation,'' \emph{arXiv preprint arXiv:1712.08500}, 2017.

\bibitem{makhdoumi}
A.~Makhdoumi, S.~Salamatian, N.~Fawaz, and M.~M{\'e}dard, ``From the
  information bottleneck to the privacy funnel,'' in \emph{2014 IEEE
  Information Theory Workshop (ITW 2014)}.\hskip 1em plus 0.5em minus
  0.4em\relax IEEE, 2014, pp. 501--505.

\bibitem{tishby}
N.~Tishby, F.~C. Pereira, and W.~Bialek, ``The information bottleneck method,''
  \emph{arXiv preprint physics/0004057}, 2000.

\bibitem{yamamoto}
H.~Yamamoto, ``A source coding problem for sources with additional outputs to
  keep secret from the receiver or wiretappers (corresp.),'' \emph{IEEE
  Transactions on Information Theory}, vol.~29, no.~6, pp. 918--923, 1983.

\bibitem{sankar}
L.~Sankar, S.~R. Rajagopalan, and H.~V. Poor, ``Utility-privacy tradeoffs in
  databases: An information-theoretic approach,'' \emph{IEEE Transactions on
  Information Forensics and Security}, vol.~8, no.~6, pp. 838--852, 2013.

\bibitem{dwork1}
C.~Dwork, F.~McSherry, K.~Nissim, and A.~Smith, ``Calibrating noise to
  sensitivity in private data analysis,'' in \emph{Theory of cryptography
  conference}.\hskip 1em plus 0.5em minus 0.4em\relax Springer, 2006, pp.
  265--284.

\bibitem{dwork2}
C.~Dwork, ``Differential privacy, in automata, languages and programming,''
  \emph{ser. Lecture Notes in Computer Scienc}, vol. 4052, p. 112, 2006.

\bibitem{oech}
Z.~Li, T.~J. Oechtering, and D.~G{\"u}nd{\"u}z, ``Privacy against a hypothesis
  testing adversary,'' \emph{IEEE Transactions on Information Forensics and
  Security}, vol.~14, no.~6, pp. 1567--1581, 2018.

\bibitem{borade}
S.~Borade and L.~Zheng, ``Euclidean information theory,'' in \emph{2008 IEEE
  International Zurich Seminar on Communications}.\hskip 1em plus 0.5em minus
  0.4em\relax IEEE, 2008, pp. 14--17.

\bibitem{huang}
S.-L. Huang and L.~Zheng, ``Linear information coupling problems,'' in
  \emph{2012 IEEE International Symposium on Information Theory
  Proceedings}.\hskip 1em plus 0.5em minus 0.4em\relax IEEE, 2012, pp.
  1029--1033.

\bibitem{huang2}
S.-L. Huang, C.~Suh, and L.~Zheng, ``Euclidean information theory of
  networks,'' \emph{IEEE Transactions on Information Theory}, vol.~61, no.~12,
  pp. 6795--6814, 2015.

\end{thebibliography}


\begin{thebibliography}{10}
\providecommand{\url}[1]{#1}
\csname url@samestyle\endcsname
\providecommand{\newblock}{\relax}
\providecommand{\bibinfo}[2]{#2}
\providecommand{\BIBentrySTDinterwordspacing}{\spaceskip=0pt\relax}
\providecommand{\BIBentryALTinterwordstretchfactor}{4}
\providecommand{\BIBentryALTinterwordspacing}{\spaceskip=\fontdimen2\font plus
\BIBentryALTinterwordstretchfactor\fontdimen3\font minus
  \fontdimen4\font\relax}
\providecommand{\BIBforeignlanguage}[2]{{%
\expandafter\ifx\csname l@#1\endcsname\relax
\typeout{** WARNING: IEEEtran.bst: No hyphenation pattern has been}%
\typeout{** loaded for the language `#1'. Using the pattern for}%
\typeout{** the default language instead.}%
\else
\language=\csname l@#1\endcsname
\fi
#2}}
\providecommand{\BIBdecl}{\relax}
\BIBdecl

\bibitem{makhdoumi}
A.~Makhdoumi, S.~Salamatian, N.~Fawaz, and M.~M{\'e}dard, ``From the
  information bottleneck to the privacy funnel,'' in \emph{2014 IEEE
  Information Theory Workshop}, 2014, pp. 501--505.

\bibitem{issa}
I.~{Issa}, S.~{Kamath}, and A.~B. {Wagner}, ``An operational measure of
  information leakage,'' in \emph{2016 Annual Conference on Information Science
  and Systems}, March 2016, pp. 234--239.

\bibitem{Calmon2}
H.~{Wang}, L.~{Vo}, F.~P. {Calmon}, M.~{M\'{e}dard}, K.~R. {Duffy}, and
  M.~{Varia}, ``Privacy with estimation guarantees,'' \emph{IEEE Transactions
  on Information Theory}, vol.~65, no.~12, pp. 8025--8042, Dec 2019.

\bibitem{yamamoto}
H.~Yamamoto, ``A source coding problem for sources with additional outputs to
  keep secret from the receiver or wiretappers (corresp.),'' \emph{IEEE
  Transactions on Information Theory}, vol.~29, no.~6, pp. 918--923, 1983.

\bibitem{sankar}
L.~Sankar, S.~R. Rajagopalan, and H.~V. Poor, ``Utility-privacy tradeoffs in
  databases: An information-theoretic approach,'' \emph{IEEE Transactions on
  Information Forensics and Security}, vol.~8, no.~6, pp. 838--852, 2013.

\bibitem{borz}
B.~{Rassouli} and D.~{G\"{u}nd\"{u}z}, ``On perfect privacy,'' \emph{IEEE
  Journal on Selected Areas in Information Theory}, vol.~2, no.~1, pp.
  177--191, 2021.

\bibitem{gun}
S.~{Sreekumar} and D.~{G\"{u}nd\"{u}z}, ``Optimal privacy-utility trade-off
  under a rate constraint,'' in \emph{2019 IEEE International Symposium on
  Information Theory}, July 2019, pp. 2159--2163.

\bibitem{khodam}
A.~Zamani, T.~J. Oechtering, and M.~Skoglund, ``A design framework for strongly
  $\chi^2$-private data disclosure,'' \emph{IEEE Transactions on Information
  Forensics and Security}, vol.~16, pp. 2312--2325, 2021.

\bibitem{Khodam22}
{A. Zamani, T. J. Oechtering, and M. Skoglund}, ``Data disclosure with non-zero
  leakage and non-invertible leakage matrix,'' \emph{IEEE Transactions on
  Information Forensics and Security}, vol.~17, pp. 165--179, 2022.

\bibitem{kostala}
Y.~Y. Shkel, R.~S. Blum, and H.~V. Poor, ``Secrecy by design with applications
  to privacy and compression,'' \emph{IEEE Transactions on Information Theory},
  vol.~67, no.~2, pp. 824--843, 2021.

\bibitem{dwork1}
C.~Dwork, F.~McSherry, K.~Nissim, and A.~Smith, ``Calibrating noise to
  sensitivity in private data analysis,'' in \emph{Theory of cryptography
  conference}.\hskip 1em plus 0.5em minus 0.4em\relax Springer, 2006, pp.
  265--284.

\bibitem{calmon4}
F.~P. {Calmon}, A.~{Makhdoumi}, M.~{Medard}, M.~{Varia}, M.~{Christiansen}, and
  K.~R. {Duffy}, ``Principal inertia components and applications,'' \emph{IEEE
  Transactions on Information Theory}, vol.~63, no.~8, pp. 5011--5038, Aug
  2017.

\bibitem{issajoon}
I.~Issa, A.~B. Wagner, and S.~Kamath, ``An operational approach to information
  leakage,'' \emph{IEEE Transactions on Information Theory}, vol.~66, no.~3,
  pp. 1625--1657, 2020.

\bibitem{asoo}
S.~Asoodeh, M.~Diaz, F.~Alajaji, and T.~Linder, ``Estimation efficiency under
  privacy constraints,'' \emph{IEEE Transactions on Information Theory},
  vol.~65, no.~3, pp. 1512--1534, 2019.

\bibitem{Total}
B.~Rassouli and D.~{G\"{u}nd\"{u}z}, ``Optimal utility-privacy trade-off with
  total variation distance as a privacy measure,'' \emph{IEEE Transactions on
  Information Forensics and Security}, vol.~15, pp. 594--603, 2020.

\bibitem{issa2}
I.~Issa, S.~Kamath, and A.~B. Wagner, ``Maximal leakage minimization for the
  shannon cipher system,'' in \emph{2016 IEEE International Symposium on
  Information Theory}, 2016, pp. 520--524.

\bibitem{zamani2022bounds}
A.~Zamani, T.~J. Oechtering, and M.~Skoglund, ``Bounds for privacy-utility
  trade-off with non-zero leakage,'' \emph{arXiv preprint arXiv:2201.08738},
  2022.

\bibitem{kosenaz}
E.~Erdemir, P.~L. Dragotti, and D.~G{\"u}nd{\"u}z, ``Active privacy-utility
  trade-off against inference in time-series data sharing,'' \emph{arXiv
  preprint arXiv:2202.05833}, 2022.

\bibitem{kosnane}
C.~T. Li and A.~El~Gamal, ``Strong functional representation lemma and
  applications to coding theorems,'' \emph{IEEE Transactions on Information
  Theory}, vol.~64, no.~11, pp. 6967--6978, 2018.

\bibitem{Calmon1}
F.~P. {Calmon} and N.~{Fawaz}, ``Privacy against statistical inference,'' in
  \emph{2012 50th Annual Allerton Conference on Communication, Control, and
  Computing}, Oct 2012, pp. 1401--1408.

\bibitem{polyanskiy2014lecture}
Y.~Polyanskiy and Y.~Wu, ``Lecture notes on information theory,'' \emph{Lecture
  Notes for ECE563 (UIUC) and}, vol.~6, no. 2012-2016, p.~7, 2014.

\bibitem{verdu}
I.~Sason and S.~Verd{\'u}, ``$f$ -divergence inequalities,'' \emph{IEEE
  Transactions on Information Theory}, vol.~62, no.~11, pp. 5973--6006, 2016.

\end{thebibliography}
\end{document}